\documentclass[11pt,reqno,tbtags]{amsart}
\usepackage[]{amsmath,amssymb,amsfonts,latexsym,amsthm,enumerate}
\usepackage{amssymb,bbm}
\usepackage{enumerate,graphicx,subfigure,paralist}
\usepackage[numeric,initials,nobysame]{amsrefs}

\numberwithin{equation}{section}

\newtheorem{maintheorem}{Theorem}

\newtheorem{theorem}{Theorem}[section]
\newtheorem*{theorem*}{Theorem}
\newtheorem{conjecture}[theorem]{Conjecture}
\newtheorem*{conjecture*}{Conjecture}
\newtheorem{lemma}[theorem]{Lemma}

\newtheorem{condition}[theorem]{Condition}

\theoremstyle{definition}{

\newtheorem*{definition*}{Definition}

}

\theoremstyle{remark}{

\newtheorem*{remark*}{Remark}

}

\newcommand{\E}{\mathbb{E}}
\renewcommand{\P}{\mathbb{P}}

\renewcommand{\epsilon}{\varepsilon}
\renewcommand{\phi}{\varphi}

\newcommand{\gt}{\tilde{G}}
\newcommand{\dtree}{\hat{\mathbbm{T}}_d}
\newcommand{\hmu}{\hat{\mu}}
\newcommand{\qt}{\tilde{Q}}
\date{}

\begin{document}
\title{Computational Transition at the Uniqueness Threshold}

\author{Allan Sly}
\address{Allan Sly\hfill\break
Microsoft Research\\
One Microsoft Way\\
Redmond, WA 98052, USA.}
\email{allansly@microsoft.com}
\urladdr{}

\begin{abstract}
The hardcore model is a model of lattice gas systems which has received  much attention in statistical physics, probability theory and theoretical computer science. It is the probability  distribution over independent sets $I$ of a graph weighted proportionally to $\lambda^{|I|}$ with fugacity parameter $\lambda$.  We  prove that at the uniqueness threshold of the hardcore model on the $d$-regular tree, approximating the partition function becomes computationally hard on graphs of maximum degree~$d$.

Specifically, we show that unless NP$=$RP there is no polynomial time approximation scheme for the partition function (the sum of such weighted independent sets) on graphs of maximum degree $d$ for fugacity $\lambda_c(d) < \lambda < \lambda_c(d) + \varepsilon(d)$ where
\[
\lambda_c = \frac{(d-1)^{d-1}}{(d-2)^d}
\]
is the uniqueness threshold on the $d$-regular tree and $\varepsilon(d)>0$ is a positive constant.  Weitz~\cite{Weitz:06} produced an FPTAS for approximating the partition function when $0<\lambda < \lambda_c(d)$ so this result demonstrates that the computational threshold exactly coincides with the statistical physics phase transition thus confirming the main conjecture of~\cite{MWW:09}.
We further analyze the special case of $\lambda=1, d=6$ and show there is no polynomial time approximation scheme for approximately counting independent sets on graphs of maximum degree $d= 6$, which is optimal, improving the previous bound of $d= 24$.

Our proof is based on specially constructed random bi-partite graphs which act as  gadgets in a reduction to MAX-CUT.
Building on the involved second moment method analysis of \cite{MWW:09} and combined with an analysis of the reconstruction problem on the tree our proof establishes a strong version of ``replica'' method heuristics developed by theoretical physicists.  The result establishes the first rigorous correspondence between the hardness of approximate counting and sampling with statistical physics phase transitions.
\end{abstract}

\maketitle

\vspace{-1cm}

\pagebreak

\section{Introduction}\label{sec:intro}

The hardcore model  is a model from statistical physics representing hardcore interaction of gas particles.  It is a probability distribution on independent sets $I$ of a graph weighted as $\frac1Z \lambda^{|I|}$ where $\lambda$ is a positive parameter called the fugacity and $Z$ is a normalizing constant called the partition function.  Physicists and probabilists  have done extensive work towards identifying the phase transitions and other properties of the model.

In computational complexity approximately counting (weighted) independent sets is a central problem.  The hardcore model is of key importance as this is exactly the problem of producing an FPRAS (fully polynomial randomized approximation scheme) for $Z$, the partition function.  When $\lambda$ is small the hardcore model has rapid decay of correlations and the partition function can be approximated either using MCMC or through computational tree methods~\cite{Weitz:06}.  For  larger fugacities long range dependencies may appear and the problem is known to be hard  when $\lambda$ is sufficiently large.

In this paper we determine a computational threshold where approximating $Z$ becomes hard.  Using an ingenious computational tree approach Weitz~\cite{Weitz:06} produced a PTAS for approximating $Z$ when $\lambda < \lambda_c(d)$ where
\[
\lambda_c(d)=  \frac{(d-1)^{d-1}}{(d-2)^d}
\]
is the uniqueness threshold for the hardcore model on the infinite $d$-regular tree~\cite{Kelly:85}, the point at which long range dependencies become possible.  Mossel, Weitz and Wormald~\cite{MWW:09} showed that beyond this phase transition local MCMC algorithms fail and conjectured that it gives the threshold for computations hardness.  While such statistical physics phase transitions are believed to coincide with the transition in computational hardness of approximating the partition function for a number of important models no such examples had been proven.  Our main result essentially confirms the conjecture of~\cite{MWW:09} giving the first such rigorous example.

\begin{maintheorem}\label{t:main}
For every $d\geq 3$ there exists  $\varepsilon(d)>0$ such that when $\lambda_c(d) < \lambda < \lambda_c(d) + \varepsilon(d)$, unless NP$=$RP, there does not exist an FPRAS for the partition function of the hardcore model with fugacity $\lambda$ for graphs of maximum degree at most $d$.
\end{maintheorem}

While we believe the result holds for all $\lambda>\lambda_c$, for technical reasons (specifically showing that an explicit function of three variables attains its maximum at a prescribed location, see Section~\ref{s:techCond} for details) the result is limited to $\lambda$ close to criticality.  This limitation notwithstanding, it clearly demonstrates the central role played by the uniqueness threshold.

When $\lambda=1$ the hardcore model is simply the uniform distribution over independent sets and the partition function is simply the number of independent sets and as such this case is of particular interest.  When $d \leq 5$  Weitz's result provides a FPRAS as $\lambda_c(d)>1$.  Conversely it is known that with $d\geq 25$ the problem is computationally hard~\cite{DFJ:02}.  While the case $d=6,\lambda=1$ does not fall within the scope of Theorem~\ref{t:main}, using a computer assisted proof, we establish the necessary technical condition and prove the following result.

\begin{maintheorem}\label{t:lambda1}
Unless NP$=$RP for every $d\geq 6$ there does not exist a fully polynomial
approximation scheme for counting independent sets on graphs of maximum degree at most $d$.
\end{maintheorem}


\subsection{Background and Previous Results}
Even on graphs of maximum degree 3 the problem of exactly counting independent sets is $\#$P hard~\cite{Greenhill:02} and as such one can at most ask when it is possible to approximately count independent sets, that is when an FPRAS exists.  As the model is self-reducible, approximate counting is equivalent to approximately sampling from the partition function~\cite{JerSin:89}.  This has led to a major line of research in analyzing the performance of MCMC techniques, particularly the Glauber dynamics.

When $\lambda\leq \frac{2}{d-2}$ the Glauber dynamics mixes rapidly~\cite{LubVig:99} which in particular gives an FPRAS for counting independent sets on graphs of maximum degree at most 4 (see \cite{DyGr:00} for similar bounds).  Weitz~\cite{Weitz:06} showed that the hardcore model has a decay of correlation property called strong spatial mixing whenever $\lambda < \lambda_c$ which implies rapid mixing on graphs of sub-exponential growth.  Moreover, his paper gives a deterministic polynomial time approximation scheme on all graphs when $\lambda < \lambda_c$ through a computational tree approximation.

Finding the ground state of the hardcore model, the largest independent set, is of course a canonical NP-hard problem and is hard to approximate even on regular graphs of degree 3~\cite{BerFuj:99}.  Intuitively the problem of counting becomes harder as $\lambda$ grows as this places more mass on the larger, harder to find, independent sets and indeed such hardness results have been established.  In~\cite{LubVig:99} it was shown that there is no FPRAS (assuming NP$=$RP) when $\lambda\leq c/d$ for $c\approx 10000$.  In the case of $\lambda=1$ this was improved to $d\geq 25$ in~\cite{DFJ:02} using random regular bi-partite graphs as basic gadgets in a hardness reduction.  They further showed that with high probability the mixing time of the Glauber dynamics on a random bipartite $d$-regular graph is exponential in the size of the graph.  Calculations of~\cite{DFJ:02} led the authors there to speculate that $\lambda_c$ may be the threshold for hardness but the evidence was not conclusive enough to make such a conjecture.

\subsubsection{Replica Heuristics}

The replica and cavity methods and heuristics  have provided powerful tools (often non-rigorous) in the study of a wide range of random optimization problems and predictions for the behavior of spin glasses and dilute mean fields spin systems~\cites{MPV:87,MezMon:09}.  Developed by theoretical physcicits, in in some cases these heuristics have been made rigorous, notably the SK model~\cite{Talagrand:06}, solution space of solutions to random constraint satisfaction problems~\cite{AchCoj:08} and the assignment problem~\cite{Aldous:01}.  In dilute spin glass models such methods have given rise to powerful new algorithms such as survey propagation (see e.g. \cite{KMRSZ:07}).

Random regular bi-partite graphs are widely known to be locally tree-like with only a small number of short cycles.  The statistical physics theory makes the following predictions for the hardcore model on typical random bi-partite $d$-regulars.  The first is that the model is expected to exhibit spontaneous symmetry breaking for $\lambda>\lambda_c$.  When $\lambda<\lambda_c$ correlations decay exponentially and the configuration (independent set) is essentially balanced between the two halves of the bi-partite graph.  By contrast when $\lambda > \lambda_c$ the configuration separates its mass unevenly placing $\Omega(n)$ more mass on one side or the other.  Configurations with a roughly equal proportion of sites on each side make up only an exponentially small fraction of the distribution.  This is intuitively plausible as the largest bi-partite sets will be those containing most of one side of the graph or the other.

The second is that this symmetry breaking splits the configuration space into two ``pure states'' of roughly equal probability.  We will denote the ``phase'' of the configuration as the side of the graph with more sites.  Conditional on the phase the spins of randomly chosen vertices are assumed to be asymptotically independent and the local neighbourhood of the configuration are given by extremal measures.  This conditional independence is a crucial element of cavity-method type arguments.

A first moment analysis of~\cite{DFJ:02} suggested that configurations obey the first prediction but their proof proceeded without specifically proving it.  In a technical tour de force the prediction was rigorously established for $\lambda_c(d) < \lambda < \lambda_c(d) + \varepsilon(d)$ in~\cite{MWW:09} using an involved second moment method analysis together with the small graph conditioning method.
The restriction to the region $\lambda < \lambda_c(d) + \varepsilon(d)$ is somewhat surprising at first as the problem ought to become easier as $\lambda$ grows.  It is the result of a technical difficulty in estimating the second moment bound. Even establishing this for $\lambda$ close to the critical value  took up fully a third of the proof.  As a central part of our proof is a modification of this method the same restriction applies.

Based on establishing the symmetry breaking~\cite{MWW:09} showed that any local reversible Markov Chain has mixing time exponential in the number of vertices by establishing a bottleneck in the mixing on asymptotically almost all random $d$-regular bi-partite graphs.  This bound is tight as subsequent results~\cite{MosSly:09}*{Theorem 4} imply rapid mixing on almost all random bi-partite graphs when $\lambda<\lambda_c(d)$.  Based on these finding they made the following conjecture.

\begin{conjecture}\label{c:cong}(\cite{MWW:09})
Unless NP$=$RP for every $d\geq 4$ and $\lambda_c(d) < \lambda$ there does not exist a fully polynomial
approximation scheme for the partition function of the hardcore model with fugacity $\lambda$ for graphs of maximum degree at most $d$.
\end{conjecture}

Phase transitions of spin systems have been known to exactly determine the region of rapid mixing in a number of systems including the ferromagnetic Ising model on $\mathbbm{Z}^2$~\cite{MarOli:94} and on the $d$-regular tree~\cite{BKMP:05}.  The first such example on completely general bounded degree graphs was recently established by Mossel and the present author~\cite{MosSly:09} showing rapid mixing of the Glauber dynamics of the ferromagnetic Ising model on graphs of maximum degree $d$ when $(d-1)\tanh \beta<1$.  The threshold $(d-1)\tanh \beta = 1$ is a statistical physics phase transition, the uniqueness threshold for the Ising model on the $d$-regular tree.

Slow mixing of MCMC algorithms do not by themselves imply hardness of approximating the partition function.  Indeed, in the ferromagnetic Ising model the mixing time of local reversible Markov chains may be exponential but nonetheless there is an FPRAS by the famous algorithm of Jerrum and Sinclair~\cite{JerSin:90}.  However, unlike the hardcore model or indeed the anti-ferromagnetic Ising model which do exhibit phase transitions, the ground states of the ferromagnetic Ising model are trivially found.

While phase transitions exists on many infinite graphs, it is the uniqueness threshold on the tree that appears to determine the onset of computational hardness in general graphs in a number of models as they represent the extreme case for correlation decay in graphs for many models.
Sokal~\cite{Sokal:01} conjectured that uniqueness on the $d$-regular tree for the hardcore model implies uniqueness on any graph of maximum degree $d$.  This conjecture was established in~\cite{Weitz:06} which further showed that for \emph{any} 2-spin system strong spatial mixing on the $d$-regular tree implies strong spatial mixing on all graphs of maximum degree $d$.  Indeed for most, although not all, spin systems the regular tree is expected to be the limiting case for extreme correlations amongst all graphs of maximum degree $d$ (see e.g. \cite{Sly:08} for more details).  The emergence of long range correlations appears to be a necessary prerequisite for hardness of sampling and this motivates the conjectures that the uniqueness threshold on the tree determines the onset of computational hardness.

In this paper we establish a form of the second heuristic prediction on a modified random bipartite graph. We show that on a polynomial sized set of vertices the spins are close to a product measure, conditional on the phase in the $L^\infty$ distance on measures.  Being able to treat large numbers of vertices as conditionally independent given the phase plays a key role in our reduction.  While some results of this nature have been established previously (see e.g. \cite{DemMon:10,MMS:09}) this is the first example we are aware of where the number of conditionally independent sites grows polynomially in the size of the graph.

\subsection{Proof Techniques}

Following the approach of~\cite{DFJ:02} and as suggested in~\cite{MWW:09} we utilize random bi-partite graphs as basic gadgets in a hardness reduction.  In those papers the basic unit of the construction is the random $d$-regular bipartite graph.  To obtain a sharp result we cannot afford to add edges to such graphs (creating degree $d+1$ vertices) so our basic gadgets are bi-partite random graphs, most of whose vertices are degree $d$ but with a small number of degree $d-1$ vertices which are used to connect to other gadgets.

We begin by constructing a graph $\gt$ which is a random bipartite graph with $n$ vertices of degree~$d$ and $m' \approx n^{\theta+\psi}$ vertices of degree $d-1$ where $\theta,\psi$ are small positive parameters.  We label the sides as ``plus'' and ``minus''  and edges are chosen according to random matchings of the vertices on the two sides.  We denote the phase of the configuration (the random independent set) to be plus or minus according to the side which has more elements of the set amongst the degree $d$ vertices.

With $U$ denoting the set of vertices of degree $(d-1)$ we consider the random partition functions $Z^\pm(\eta)$ giving the sum over $\lambda^{|\sigma|}$ over all configurations with phase $\pm$ and with $\sigma_U=\eta$ where $\eta\in\{0,1\}^U$.  We show that in expectation the $\E Z^\pm(\eta)$ are essentially proportional to the probabilities of a product measure on $U$ whose marignals are given by the marginals of extremal Gibbs measures for the hardcore model on the $(d-1)$-ary tree.  Our proof requires that this holds approximately for the $Z^\pm(\eta)$ themselves and adopt the second moment approach of~\cite{MWW:09} including their use of the small graph conditioning method~\cite{Wormald:99}.  While still involved, by estimating ratios of quantities in our model to quantities calculated in~\cite{MWW:09} we greatly simplify these computations.  We are, however, still left with the same technical condition as~\cite{MWW:09} which we describe in the next subsection.

Even this approximate conditional independence is not sufficient for our reduction.  To this end we construct a new random graph $G$ by appending $(d-1)$-ary trees of height $\psi \log_{d-1} n$ onto $U$ and denote the set of $m\approx n^\theta$ roots of the trees as $V$ which are of degree $d-1$. Our proof proceeds to show that, conditional on the phase, $\sigma_V$ is very close to a product measure.  We note that appending the trees reweights the probabilities on configurations $\sigma_U$ but it does so in a quantifiable way.

By construction the spins $\sigma_V$ are conditionally independent given $\sigma_U$.  Moreover, the statistical physics heuristics imply that the configuration of the neighbourhood around $\sigma_V$ should be given by an extremal semi-translation invariant Gibbs measure on the tree with strong decay of correlation from the root to the leaves of the tree.  Based on this intuition, we show that after conditioning on the phase the probability that $\sigma_U$ has a non-negligible influence on $\sigma_V$ is doubly exponential small in the height.  Through this we can establish its distribution with bounds in the $L^\infty$ norm.  This is done by bounding the probability that the spins in a distant level influence the root using methods from the ``reconstruction problem on the tree'' (see e.g.~\cite{Mossel:04,Sly:09}).

The random graph $G$ constitutes our gadget.  Given a graph $H$ on up to $n^{\theta/4}$ vertices we construct $H^G$ by taking a copy of $G$ for each vertex of $H$.  Then for every edge  in $H$ we connect $n^{3\theta/4}$ vertices between each side of $V$ in the corresponding copies of $G$ maintaining the maximum degree $d$.  Since the spins in $V$ are almost conditionally independent given the phase we can estimate the effect of adding these edges.  An easy calculation shows that the most efficient arrangement is to have connected gadgets have opposite phases. The hardcore model on $H^G$ puts most of its mass on configurations whose phases are solutions to MAX-CUT on $H$.  Hence, by the equivalence of approximate counting and approximate sampling, this gives a randomized reduction to MAX-CUT.

\subsection{Preliminaries}\label{s:preliminaries}
For a finite graph $G$ with edge set $E(G)$ the independent sets are subsets of the vertices containing no adjacent vertices or equivalently elements of the set of configurations
\[
I(G) = \{\sigma \in [0,1]^G:\forall (u,v)\in E(G), \sigma_u \sigma_v = 0\}.
\]
The \emph{Hardcore Model} is a probability distribution over independent sets of a graph $G$ defined by
\begin{equation}\label{e:defnHardcore}
\P_G(\sigma)=\frac1{Z_G(\lambda)} \lambda^{\sum_{v\in G}\sigma_v} \mathbbm{1}_{\sigma\in I(G)}
\end{equation}
where $Z_G(\lambda)=\sum_{\sigma\in I(G)} \lambda^{|\sigma|}=\sum_{\sigma\in I(G)} \lambda^{\sum_{v\in G}\sigma_v}$ is a normalizing constant known as the \emph{partition function} and is a weighted counting of the independent sets.  When $\lambda=1$ the hardcore model is the uniform measure on independent sets and $Z_G(1)$ is the number of independent sets of the graph.

The definition of the hardcore model can be extended to infinite graphs by way of the DLR condition which essentially says that for every finite set $A$ the configuration on $A$ is given by the Gibbs distribution given by a random boundary generated by the measure outside of $A$.  Such a measure is called a Gibbs measure and there may be more one or infinitely many such measures (see e.g. \cite{Georgii:88} for more details).  When there is exactly one Gibbs measure we say the model has \emph{uniqueness}. Our main result relates the uniqueness threshold on $\mathbbm{T}_d$, the infinite $d$-regular tree, to the hardness of approximating the partition function on graphs of maximum degree~$d$.

The hardcore model on $\mathbbm{T}_d$ undergoes a phase transition at $\lambda_c(d)=\frac{(d-1)^{d-1}}{(d-2)^d}$
with uniqueness when $\lambda\leq \lambda_c$ and non-uniqueness when $\lambda > \lambda_c$~\cite{Kelly:85}.   The following picture is described in \cite{MWW:09}.  For every $\lambda$ there exists a unique translation invariant Gibbs measure $\mu=\mu_{d,\lambda}$ known as the free measure with occupation density $p^*=\mu(\sigma_\rho)$ for $\rho$ the root of the tree.  When $\lambda>\lambda_c$ there also exist two semi-translation invariant (that is invariant under parity preserving automorphisms of $\mathbbm{T}_d$) measures $\mu_+$ and $\mu_-$ whose occupation densities we denote by $p^+=\mu_+(\sigma_\rho),p^-=\mu_-(\sigma_\rho)$.  These measures are obtained by conditioning on level $2\ell$ (resp. $2\ell+1$) of the tree to be completely occupied and taking the weak limit as $\ell\to\infty$.

It will also be of use to discuss related measures on the infinite $(d-1)$-ary tree $\hat{\mathbbm{T}}^d$ rooted at $\rho$.  We define analogously the measures $\hmu_+$ and $\hmu_-$ obtained by conditioning on level $2\ell$ (resp. $2\ell+1$) of $\dtree$ to be completely occupied and taking the weak limit as $\ell\to\infty$.  We set $q^+$ and $q^-$ to be the respective occupation densities $q^+=\hmu_+(\sigma_\rho),q^-=\hmu_-(\sigma_\rho)$ of the root $\rho$.

The measure $\mu_\pm$ and $\hmu_\pm$ are naturally related as follows.  Let $v$ be a child of $\rho$ and denote $\mathbbm{T}_v$ to be the subtree of $\mathbbm{T}^d$ rooted at $v$.  There is a natural identification of $\mathbbm{T}^d \setminus \mathbbm{T}_v$ with the $(d-1)$-ary tree $\hat{\mathbbm{T}}^d$ and under this identification the measures satisfy
\begin{equation}\label{e:measureIdentification}
\hmu_\pm(\sigma\in \cdot) = \mu_\pm(\sigma_{\mathbbm{T}^d \setminus \mathbbm{T}_v}\in\cdot|\sigma_v=0).
\end{equation}
In particular since $\sigma_\rho=1$ implies  $\sigma_v=0$ for an independent set  in $\mathbbm{T}^d$ it follows that
\begin{equation}\label{e:pqRelation}
q^\pm = \frac{p^\pm}{1-p^\mp}.
\end{equation}
Furthermore, standard tree recursions for Gibbs measures (see e.g. \cite{MWW:09}) establish that
\[
q^\pm = \frac{\lambda(1-q^\mp)^{d-1}}{1+\lambda(1-q^\mp)^{d-1}}
\]
and consequently by equation~\eqref{e:pqRelation},
\begin{equation}\label{e:qTreeRelation}
\frac{q^\pm}{1-q^\pm} = \lambda(1-q^\mp)^{d-1} = \lambda\left(\frac{1-p^\pm - p^\mp}{1-p^\pm}\right)^{d-1}.
\end{equation}
It is shown in~\cite{MWW:09}*{Section 4} and~\cite{DFJ:02}*{Claim 2.2} that the following hold for $\lambda>\lambda_c$:
\begin{enumerate}
\item
The solutions to  $h(\alpha)=\beta,h(\beta)=\alpha$ with $(\alpha,\beta)\in \mathcal{T}=\{(\alpha,\beta):\alpha,\beta\geq 0, \alpha+\beta\leq 1 \}$ where
\[
h(x)=(1-x)\left[1-\left(\frac{x}{\lambda(1-x)}\right)^{1/d} \right]
\]
are exactly $(p^+,p^-)$, $(p^-,p^+)$ and $(p^*,p^*)$.  These densities satisfy  $p^- < p^* < p^+$ and when $\lambda \downarrow \lambda_c$ we have that $p^*,p^+,p^-\to 1/d$.
\item The points $(p^+,p^-)$ and $(p^-,p^+)$ are the maxima of $\Phi_1(\alpha,\beta)$ in $\mathcal{T}$ where
\begin{align*}
\Phi_1(\alpha,\beta)&=(\alpha+\beta)\log \lambda -\alpha\log\alpha-\beta\log\beta-d(1-\alpha-\beta)\log(1-\alpha-\beta)\\
& \qquad + (d-1)\left((1-\alpha)\log(1-\alpha) + (1-\beta)\log(1-\beta) \right).
\end{align*}
\end{enumerate}

\subsubsection{Technical Conditions}\label{s:techCond}
We now describe the technical condition necessary for our result.
The function in question is
\begin{align}
&f(\alpha,\beta,\gamma,\delta,\epsilon) = 2(\alpha+\beta)\log \lambda + H(\alpha)+H_1(\gamma,\alpha)+H_1(\alpha-\gamma,1-\alpha) + H(\beta) + H_1(\delta,\beta) \nonumber\\
&\qquad + H_1(\beta-\delta,1-\beta) + d\bigg[ H_1(\gamma,1-2\beta+\delta)-H(\gamma)+H_1(\epsilon,1-2\beta+\delta-\gamma) \nonumber\\
&\qquad +H_1(\alpha-\gamma-\epsilon,\beta-\delta)-H_1(\alpha-\gamma,1-\gamma) + H_1(\alpha-\gamma,1-\beta-\gamma-\epsilon) -H_1(\alpha-\gamma,1-\alpha)\bigg]
\end{align}
where $H_1(x,y) = - x(\log x - \log y) + (x-y)(\log(y-x)-\log(y))$ and $H(x)=H(x,1)$
and where $f$ is defined in the range $(\alpha,\beta)\in\mathcal{T}$ and
\begin{equation}\label{e:greekRegion}
\alpha-\gamma-\epsilon\geq 0,\beta-\delta\geq 0, 1-2\beta+\delta-\gamma-\epsilon\geq 0.
\end{equation}
which emerges naturally when calculating the second moment of the partition function.
\begin{condition}\label{cond:technical}
The technical condition is that there exists a constant $\chi>0$ such that when when $|p^+-\beta|,|p^- - \alpha|<\chi$ the function $g_{\alpha,\beta}(\gamma,\delta,\epsilon)=f(\alpha,\beta,\gamma,\delta,\epsilon)$ attains its unique maximum in the set~\eqref{e:greekRegion} at the point $(\gamma^*,\delta^*,\epsilon^*)=(\alpha^2,\beta^2,\alpha(1-\alpha-\beta))$.
\end{condition}
The following result of~\cite{MWW:09} establishes Condition~\ref{cond:technical} when $\lambda_c<\lambda<\lambda_c(d)+\epsilon(d)$.
\begin{lemma}[\cite{MWW:09}*{Lemma 6.10, Lemma 5.1}]\label{l:MWWfMax}
For each $d\geq3$ there exists $\chi>0$ such that when $|\alpha-\frac1d|,|\beta-\frac1d| < \chi$ then $g_{\alpha,\beta}(\gamma,\delta,\epsilon)$ has a unique maximum at $(\gamma^*,\delta^*,\epsilon^*)$ where $\gamma^*=\alpha^2,\delta^*=\beta^2,\epsilon^*=\alpha(1-\alpha-\beta)$.
\end{lemma}

In Section~\ref{s:computerAssit} we give a computer assisted proof which establishes Condition~\ref{cond:technical} in the special case of $\lambda=1$ and $d=6$.
Two other technical conditions we make use of in the proof are that
\begin{equation}\label{e:extraConditions}
q^+ q^-  (d-1)< 1 \quad \hbox{ and } q^+<\frac35.
\end{equation}
Both  conditions holds in the regions of interest as we have that $q^+,q^- \to \frac1{d-1}$ when $\lambda \downarrow \lambda_c$ and $q^+\approx 0.423 , q^- \approx 0.056$ when $\lambda=1$ and $d=6$.  The first can be shown to hold for all $\lambda>\lambda_c$ with a somewhat involved proof while the latter is unnecessary but somewhat simplifies the proof.

\subsection{Comments and Open Problems}
The main open problem, of course, is to remove the  $\lambda < \lambda_c(d) + \varepsilon(d)$ condition, ideally with a proof avoiding the second moment analysis.  Alternatively, one could try and establish the Condition~\ref{cond:technical} for all $\lambda > \lambda_c$ and $d\geq 3$.

Another natural problem  is to establish the correspondence between computational hardness and phase transitions in the anti-ferromagnetic Ising model. While calculations of the style of~\cite{MWW:09} are not available and are likely to be even more challenging, it may be possible to avoid them.  Indeed results of~\cite{MMS:09} already imply conditional local weak convergence of the configuration but not in a strong enough form to complete necessary reduction.

In Section~\ref{s:constructG} we detail the construction for $G$ and show how it. In Section~\ref{s:partition} we analyze the first and second moments of the partition functions $Z^\pm(\eta)$.  In Section~\ref{s:reconstruction} we analyse the reconstruction problem on the tree and establish the conditional distributions of $\sigma_V$.  Finally in Section~\ref{s:computerAssit} we sketch the computer assisted proof that Condition~\ref{cond:technical} holds when $d=6$ and $\lambda=1$.

\subsection{Acknowledgements}

A.S. would like to thank Elchanan Mossel for his generous encouragement, guidance, support and advice with this project and also Dror Weitz for helpful discussions.  The worked was initiated when the A.S. was a student at UC Berkeley where he was supported by NSF CAREER grant DMS-0548249 and by DOD ONR grant (N0014-07-1-05-06) 1300/08.

\section{Proof of Theorem~\ref{t:main} and~\ref{t:lambda1}}

In this section we first describe the construction of our base random graph $G$ which will be the basic gadget in our reduction.  We state a theorem describing the properties of the hardcore model on $G$ and then proceed to show how this establishes the reduction for Theorems~\ref{t:main} and~\ref{t:lambda1}.

\subsection{Construction of $G$}\label{s:constructG}

We begin by constructing a random bi-partite (multi)graph $\gt=\gt(n,\theta,\psi)$ where $n$ is a positive integer and $0<\theta,\psi<\frac18$ are positive constants which will be chosen to depend on $\lambda$ and $d$.  This graph will be the basis of our construction of $G$.
\begin{itemize}
\item The bipartite graph is constructed in two halves which we will call respectively the plus half and the minus half each with $n+m'$ vertices where $m'=(d-1)^{\lfloor\theta \log_{d-1} n\rfloor+2\lfloor\frac{\psi}{2}\log_{d-1}\rfloor n}$.
\item The vertices of each side are split into two sets $W^{\pm}$ and $U^{\pm}$ of size $n$ and $m'$ respectively.  We label the vertices of $U^{\pm}$ by $u_1^\pm,\ldots,u^\pm_{m'}$.
\item We connect $d-1$ edges to each vetex by taking $d-1$ random perfect matchings of  $W^+\cup U^+$ with $W^-\cup U^-$ and adding an edge between each pair of matched vertices.
\item We take one more perfect matching of $W^+$ with $W^-$ and add an edge between each pair of matched vertices.
\end{itemize}

In this construction the vertices in $W=W^{+}\cup W^-$ are of degree $d$ and the vertices in $U=U^+\cup W^-$ are of degree $d-1$.  Note that in this construction there will be multiple edges between vertices with asymptotically constant probability bounded away from 1.  However, in the hardcore model multiple edges are irrelevant and we simply treat them as single edges (some degrees will be decreased but this will not affect our proof).

We now complete our construction of $G=G(n,\theta,\psi)$ by adjoining trees onto $U^+$ and to $U^-$.
\begin{itemize}
\item Construct a collection of $m=(d-1)^{\lfloor\theta \log_{d-1}n\rfloor}$ disconnected $(d-1)$-ary trees of depth $2\lfloor\frac{\psi}{2} \log_{d-1}n\rfloor$ rooted at $v_1^+,\ldots,v^+_{m}$.  The total number of leaves of the trees is $m'$.
\item Adjoin this collection of trees to $U^+$ by identifying each vertex of $U^+$ with the leaf of one of the trees.  Denote the set of roots as $V^+$ which are vertices of degree $d-1$.
\item Perform the analogous construction on $U^-$ to complete $G$.
\end{itemize}
This construction yields a bi-partite graph of maximum degree $d$ with $m$ vertices of degree $d-1$ on each side.  We now consider a the Hardcore model $P_G(\sigma)$ on $G$.  Our construction is a modification of the model considered in \cite{MWW:09} where they showed that on a.a.a random bi-partite $d$-regular graphs the probability of ``balanced'' sets is exponentially small.  This is also the case for our construction and we define the phase of the configuration as
\[
Y=Y(\sigma):=\begin{cases}+1 &\hbox{if }\sum_{w\in W^+} \sigma_w \geq \sum_{w\in W^-} \sigma_w, \\
-1 &\hbox{if } \sum_{w\in W^+} \sigma_w < \sum_{w\in W^-} \sigma_w. \end{cases}
\]
We define the product measure $Q_V^+$ (respectively $Q^-$) on configurations on $V=V^+\cup V^-$ so that the spins are iid Bernoulli with probability $q^+$ (resp.~$q^-$) on $V^+$ and $q^-$ (resp.~$q^+$) on $V^-$, i.e.,
\[
Q_V^\pm(\sigma_V) := (q^\pm)^{\sum_{v\in V^+} \sigma_{v}} (1-q^\pm)^{m-\sum_{v\in V^+} \sigma_{v}} (q^\mp)^{\sum_{v\in V^-} \sigma_{v}} (1-q^\mp)^{m-\sum_{v\in V^-} \sigma_{v}}.
\]
We define $Q_U$ on $U=U^+\cup U^-$ similarly.  With these definitions we establish the following result about hardcore model on $G$.

\begin{theorem}\label{t:Gproperties}
For every $d\geq 3$ when $\lambda_c(d) < \lambda$ and when Condition~\ref{cond:technical} and equation \eqref{e:extraConditions} hold there exists constants $\theta(\lambda,d),\psi(\lambda,d)>0$ such that the graph $G(n,\theta,\psi)$ has $(2+o(1))n$ vertices and satisfies the following with high probability:
\begin{itemize}
\item The phases occur with roughly balanced probability so that
\begin{equation}\label{e:GpropA}
\P_G(Y=+) \geq \frac1n, \P_G(Y=-) \geq \frac1n.
\end{equation}
\item The conditional distribution of the configuration on $V$ satisfies
\begin{equation}\label{e:GpropB}
\max_{\sigma_V} \left| \frac{\P_G(\sigma_V|Y=\pm)}{Q_V^\pm(\sigma_V)} - 1\right| \leq n^{-2\theta}.
\end{equation}
\end{itemize}
\end{theorem}
The proof of this theorem is deferred to Section~\ref{s:reconstruction}.

\subsection{Reduction to Max-Cut}

We now demonstrate how to use Theorem~\ref{t:Gproperties} to establish a reduction from sampling from the hardcore model to Max-Cut.  Let $H$ be a graph on up to $\frac{1}{d-1}n^{\theta/4}$ vertices.  With a random bi-partite graph $G=G(n,\theta,\psi)$ constructed as above we define $H^G$ as follows.
\begin{itemize}
\item Take the graph comprising $|H|$ disconnected copies of $G$ and identify each copy with with a vertex in $H$  labeling the copies $(G_x)_{x\in H}$.  Denote this graph by $\widehat{H}^G$.  We let $V^+_x$ and $V^-_x$ denote the vertices of $G_x$ corresponding to $V^+$ and $V^-$.
\item For every edge $(x,y)$ in the graph $H$ add $n^{3\theta/4}$ edges between $V^+_x$ and $V^+_y$ and similarly add $n^{3\theta/4}$ edges between and $V^-_x$ and $V^-_y$.  This can be done deterministically in such a way that no vertex in $\widehat{H}^G$ has its degree increased by more than 1.  Denote the resulting graph by $H^G$.
\end{itemize}
The resulting graph has maximum degree $d$.  For each $x\in H$ we let $Y_x=Y_x(\sigma)$ denote the phase of a configuration $\sigma$ on $G_x$.  Let $\mathcal{Y}=(Y_x)_{x\in H}\in\{0,1\}^H$ denote the vector of phases of the $G_x$. Denote the partition function given the phase $\mathcal{Y}$ by
\[
Z_{H^G}(\mathcal{Y}')=\sum_{\sigma\in I(H^G)} \lambda^{|\sigma|} \mathbbm{1}\left(\mathcal{Y}(\sigma)=\mathcal{Y}' \right).
\]
\begin{lemma}\label{l:HGcutProb}
Suppose that $G$ satisfies equations~\eqref{e:GpropA} and~\eqref{e:GpropB} of Theorem~\ref{t:Gproperties}.  Then
\begin{equation}\label{e:phaseProbs}
\frac{Z_{\widehat{H}^G}(\mathcal{Y}')}{Z_{\widehat{H}^G}} = \P_G(Y=+)^{\sum_{x\in H} \mathbbm{1}_{Y_x'=+}} \cdot \P_G(Y=-)^{\sum_{x\in H} \mathbbm{1}_{Y_x'=-}} \geq  n^{-n^{\theta/4}},
\end{equation}
and
\begin{equation}\label{e:cutProb}
\frac{Z_{H^G}(\mathcal{Y}')}{Z_{\widehat{H}^G}(\mathcal{Y}')} = (C_H+o(1))\left[\frac{(1-q^+ q^-)^2} {(1-(q^+)^2)(1-(q^-)^2)}\right]^{n^{3\theta/4}\mathrm{Cut}(\mathcal{Y}')}  ,
\end{equation}
where $C_H=\left[(1-(q^+)^2)(1-(q^-)^2)\right]^{n^{3\theta/4}E(H)}$ and where
$\mathrm{Cut}(\mathcal{Y}')=\#\{(x,y)\in E(H): \mathcal{Y}_x'\neq \mathcal{Y}_y'\}$
denotes the number of edges in cut of $H$ induced by $\mathcal{Y}'$.
\end{lemma}

\begin{proof}
Since the graph $\widehat{H}^G$ consists of a collection of disconnected copies of $G$, the distribution of a configuration on $\widehat{H}^G$ is given by the product measure of configurations on the $(G_x)_{x\in H}$.  In particular the phases are independent and so
\begin{align*}
\frac{Z_{\widehat{H}^G}(\mathcal{Y}')}{Z_{\widehat{H}^G}}
=\P_{\widehat{H}^G}\left(\mathcal{Y}(\sigma)=\mathcal{Y}'\right)= \P_G(Y=+)^{\sum_{x\in H} \mathbbm{1}_{Y_x'=+}} \cdot \P_G(Y=-)^{\sum_{x\in H} \mathbbm{1}_{Y_x'=-}} \geq  n^{-n^{\theta/4}},
\end{align*}
which establishes equation~\eqref{e:phaseProbs}.  Now the ratio of the partition functions in \eqref{e:cutProb} is exactly the probability that the configuration $\sigma$ sampled under $\P_{\widehat{H}^G}$ is also an independent set for $H^G$ after adding in the extra edges, that is
\begin{align*}
\frac{Z_{H^G}(\mathcal{Y}')}{Z_{\widehat{H}^G}(\mathcal{Y}')} &= \P_{\widehat{H}^G}\left(\sigma\in I(H^G)\mid \mathcal{Y}(\sigma)=\mathcal{Y}' \right)\\
&= \P_{\widehat{H}^G}\left(\forall(v,v')\in E(H^G)\setminus E(\widehat{H}^G), \sigma_{v}\sigma_{v'}\neq 1 \mid \mathcal{Y}(\sigma)=\mathcal{Y}' \right).
\end{align*}
Now by equation~\eqref{e:GpropB}, conditional on the phase $\mathcal{Y}'$ the spins of $\sigma_{\cup_{x\in H} V_x}$ are asymptotically conditionally independent with probabilities $q^+$ or $q^-$ depending on the phase.  It follows that
\begin{align*}
&\P_{\widehat{H}^G}\left(\forall(v,v')\in E(H^G)\setminus E(\widehat{H}^G), \sigma_{v}\sigma_{v'}\neq 1 \mid \mathcal{Y}(\sigma)=\mathcal{Y}' \right)\\
&\qquad= (1+o(1)) \prod_{(v,v')\in E(H^G)} \P_{\widehat{H}^G} (\sigma_{v}\sigma_{v'}\neq 1\mid \mathcal{Y}(\sigma)=\mathcal{Y}').
\end{align*}
If $(x,x')\in E(H)$ then by direction calculations and equation~\eqref{e:GpropB}
\begin{align*}
&\prod_{v\in G_x,v'\in G_{x'}:(v,v')\in E(H^G)} \P_{\widehat{H}^G} (\sigma_{v}\sigma_{v'}\neq 1 \mid \mathcal{Y}(\sigma)=\mathcal{Y}')\\
&\qquad=\begin{cases} (1+O(n^{-\theta})) \left((1-(q^+)^2)(1-(q^-)^2)\right)^{n^{3\theta/4}} & \hbox{if } Y_x=Y_{x'},\\
(1+O(n^{-\theta})) \left((1-q^+ q^-)^2\right)^{n^{3\theta/4}} & \hbox{if } Y_x \neq Y_{x'}.\end{cases}
\end{align*}
Combining the above estimates we have that
\begin{align*}
\frac{Z_{H^G}(\mathcal{Y}')}{Z_{\widehat{H}^G}(\mathcal{Y}')} = (C_H+o(1))\left[\frac{(1-q^+ q^-)^2} {(1-(q^+)^2)(1-(q^-)^2)}\right]^{n^{3\theta/4}\mathrm{Cut}(\mathcal{Y}')},
\end{align*}
which completes the proof.
\end{proof}
Given the previous lemma we now show how to produce the randomized reduction to Max-Cut establishing Theorems~\ref{t:main} and~\ref{t:lambda1}.
\begin{proof}[Theorem~\ref{t:main} and \ref{t:lambda1}]
Let $H$ be a graph on at most $\frac1{d-1} n^{\theta}$ vertices.  Take an instance of a random graph $G=G(n,\theta,\psi)$ according to the construction in Section~\ref{s:constructG}.  By Theorem~\ref{t:Gproperties} with probability tending to 1 the graph satisfies equations~\eqref{e:GpropA} and~\eqref{e:GpropB}. Assume that it does and construct the graph $H^G$ which has at most $O(n^{1+\theta})$ vertices and maximum degree $d$.

Now suppose there exists an FPRAS for the partition function for the hardcore model with fugacity $\lambda$ on graphs of maximum degree $d$.  We now use the equivalence of approximating the partition function and approximately sampling for the hardcore model described in the introduction.  In polynomial time we may approximately sample from the hardcore model on $H^G$ to within $\delta$ of the Gibbs distribution in total-variation distance for any $\delta>0$.  Let $\sigma'$ denote such an approximate sample.  We may couple $\sigma'$ and with $\sigma$ distributed according to the Gibbs measure so that $\P(\sigma'\neq\sigma)\leq \delta$.  We now consider the phase of $\sigma$.  Let $\mathcal{Y}',\mathcal{Y}''\in\{0,1\}^{H}$ such that
\[
\mathrm{Cut}(\mathcal{Y}') > \mathrm{Cut}(\mathcal{Y}'').
\]
Then by Lemma~\ref{l:HGcutProb} we have that
\begin{align}\label{e:phaseRatio}
\frac{\P(\mathcal{Y}(\sigma)=\mathcal{Y}')}{\P(\mathcal{Y}(\sigma)=\mathcal{Y}'')}
=\frac{Z_{H^G}(\mathcal{Y}')}{Z_{H^G}(\mathcal{Y}'')}&\geq \frac{(1+o(1))Z_{\widehat{H}^G}(\mathcal{Y}')}{Z_{\widehat{H}^G}(\mathcal{Y}'')} \left[\frac{(1-q^+ q^-)^2} {(1-(q^+)^2)(1-(q^-)^2)}\right]^{n^{3\theta/4}[\mathrm{Cut}(\mathcal{Y}')-\mathrm{Cut}(\mathcal{Y}'')]}\nonumber\\
&\geq (1+o(1))n^{-n^{\theta/4}} \left[\frac{(1-q^+ q^-)^2} {(1-(q^+)^2)(1-(q^-)^2)}\right]^{n^{3\theta/4}[\mathrm{Cut}(\mathcal{Y}')-\mathrm{Cut}(\mathcal{Y}'')]}.
\end{align}
As we have that $0<q^-<q^+<1$ if follows that $(1-q^+ q^-)^2 - (1-(q^+)^2)(1-(q^-)^2) = (q^+ - q^-)^2 > 0$
and hence
\[
\frac{(1-q^+ q^-)^2} {(1-(q^+)^2)(1-(q^-)^2)} > 1.
\]
Therefore, for large enough $n$ by equation~\eqref{e:phaseRatio} it follows that
\[
\frac{\P(\mathcal{Y}(\sigma)=\mathcal{Y}')}{\P(\mathcal{Y}(\sigma)=\mathcal{Y}'')} \geq \left[\frac{(1-q^+ q^-)^2} {(1-(q^+)^2)(1-(q^-)^2)}\right]^{\frac12 n^{3\theta/4}} \geq 4^{n^{\theta/4}}.
\]
Since the size of $\{0,1\}^{|H|}$ is only $2^{n^{\theta/4}}$ it follows that with  probability at least $1-2^{|H|}$ that $\mathrm{Cut}(\mathcal{Y}(\sigma))$ attains the maximum value.  Hence with probability at least $1-\delta-o(1)$ the phases $\mathcal{Y}(\sigma')$ of the approximate sample $\sigma'$ also attains a maximum cut in $H$.  As such we have constructed a randomized polynomial-time reduction from approximating the partition function of the hardcore model to constructing a maximum cut.  It follows that unless RP$=$NP there is no polynomial-time algorithm for approximating the partition function of the hardcore model for $\lambda_c(d)<\lambda<\lambda_c(d)+\varepsilon(d)$ on graphs of maximum degree $d$ or when $\lambda=1$ on graphs of maximum degree 6 or more.
\end{proof}
\pagebreak

\section{The partition function of $\tilde{G}$}\label{s:partition}
In this section we analyse the hardcore model on the random bi-partite graph $\gt$ and in particular consider the effect of conditioning on the spins in $U=U^+ \cup U^-$.  For $\eta\in\{0,1\}^U$ we define $Z_{\gt}(\eta)$ to be the partition function over configurations whose restriction to $U$ is $\eta$, that is
\[
Z_{\gt}(\eta)=\sum_{\sigma\in I(\gt):\sigma_U=\eta} \lambda^{|\sigma|}.
\]
Our analysis borrows heavily on hard computations carried out in~\cite{MWW:09}.  There they considered a random $d$-regular bipartite graph where each side has $n$ vertices and the edges are chosen according to $d$ independent perfect matchings of the vertices of the sides.  They denote $Z^{\alpha,\beta}$ to be the weighted sum over configurations of the graph with $\alpha n$ and $\beta n$ vertices on the plus and minus sides of the configuration (for $\alpha,\beta$ such that $\alpha n,\beta n$ are integers).  We will denote their quantity by $Z^{\alpha,\beta}_{\mathrm{MWW}}$.  In the same spirit define
\[
Z_{\gt}^{\alpha,\beta}(\eta)= \sum_{\sigma:\sigma_U=\eta,\sum_{w\in W^+} \sigma_{w}=\alpha n,\sum_{w\in W^-}  \sigma_{w}=\beta n}\lambda^{|\sigma|}.
\]
\begin{lemma}\label{l:ratio1Moment}
For any $(\alpha,\beta)$ in the interior of $\mathcal{T}$ and all $\eta\in\{0,1\}^U$ we have that:
\begin{equation}
\E Z^{\alpha,\beta}_{\gt}(\eta)= (1+O(n^{-1/2}))C^*\left(\lambda \left(\frac{1-\alpha-\beta}{1-\beta} \right)^{d-1}\right)^{\eta^-}
\left(\lambda \left(\frac{1-\alpha-\beta}{1-\alpha} \right)^{d-1}\right)^{\eta^+} \E Z^{\alpha,\beta}_{\mathrm{MWW}}
\end{equation}
where
\[
C^*=\left(\frac{(1-\alpha)(1-\beta)}{1-\alpha-\beta}\right)^{m'}
\]
and where $\eta^{\pm}$ denotes $\sum_{u\in U^\pm} \eta_{u}$.
\end{lemma}
\begin{proof}
We follow the approach of~\cites{MWW:09} in estimating these quantities.  In total there are $\binom{n}{\alpha n} \binom{n}{\beta n}$ choices of configurations on $W$ with $\alpha n$ sites on the top and $\beta n$ sites on the bottom.  Then by calculating the probability that a perfect matching does not connect two 1's of the configuration we have that
\begin{equation}\label{e:gt1moment}
\E Z_{\gt}^{\alpha,\beta}(\eta) = \lambda^{\alpha n+\beta n + \eta^+ + \eta^-}  \binom{n}{\alpha n} \binom{n}{\beta n}
\left[ \frac{\binom{n+m'-\beta n -\eta^-}{\alpha n + \eta^+}}{\binom{n+m'}{\alpha n+\eta^+}} \right]^{d-1}
\frac{\binom{n-\beta n}{\alpha n}}{\binom{n}{\alpha n}},
\end{equation}
while by~\cite{MWW:09} we have that
\[
Z^{\alpha,\beta}_{\mathrm{MWW}}= \lambda^{\alpha n+\beta n}  \binom{n}{\alpha n} \binom{n}{\beta n}
\left[ \frac{\binom{n-\beta n}{\alpha n}}{\binom{n}{\alpha n}} \right]^{d}.
\]
Now since $|U| = O(n^{1/4})$ it follows from Lemma~\ref{l:bionmialPerturb} below that
\begin{align*}
\frac{\E Z_{\gt}^{\alpha,\beta}(\eta)}{\E Z^{\alpha,\beta}_{\mathrm{MWW}}}=(1+O(n^{-1/2})) C^*
\left(\lambda \left(\frac{1-\alpha-\beta}{1-\beta} \right)^{d-1}\right)^{\eta^-}
\left(\lambda \left(\frac{1-\alpha-\beta}{1-\alpha} \right)^{d-1}\right)^{\eta^+}.
\end{align*}
\end{proof}
To complete Lemma~\ref{l:ratio1Moment} we give the following lemma which a simple expansion of factorials which will use repeatedly  throughout this section.
\begin{lemma}\label{l:bionmialPerturb}
When $0<b<a$ are integers and $x^2+y^2 \leq \min\{b,a-b\}$ then
\begin{align}\label{e:bionmialPerturb}
\frac{\binom{a+x}{b+y}}{\binom{a}{b}} = \left(1+O\left(\frac{x^2+y^2}{\min\{b,a-b\}}\right)\right)
\left(\frac{a}{a-b}\right)^x \left(\frac{a-b}{b}\right)^y.
\end{align}
\end{lemma}
\begin{proof}
By expanding out factorials we have that
\begin{align*}
\frac{\binom{a+x}{b+y}}{\binom{a}{b}} &= \frac{(a+x)!}{a!}\frac{b!}{(b+y)!}\frac{(a-b)!}{(a-b+x-y)!}\nonumber\\
&= \left( a^x \prod_{i=1}^x (1+\frac{i}{a})\right)\left( b^{-y} \prod_{i=1}^y (1-\frac{i}{b}) \right)\left((a-b)^{x-y} \prod_{i=0}^{x-y} (1+\frac{i}{a-b}) \right) \\
&= \left(1+O\left(\frac{x^2+y^2}{\min\{b,a-b\}}\right)\right)
\left(\frac{a}{a-b}\right)^x \left(\frac{a-b}{b}\right)^y.
\end{align*}
\end{proof}

We now sum over $(\alpha,\beta)$ and define the conditional partition functions as
\begin{align*}
Z_{\gt}^{+}(\eta) &= \sum_{\alpha\geq \beta} Z_{\gt}^{\alpha,\beta}(\eta)= \sum_{\sigma:\sigma_U=\eta,\sum_{w\in W^+} \sigma_{w}\geq \sum_{w\in W^-} \sigma_{w}}\lambda^{|\sigma|}\\
Z_{\gt}^{-}(\eta) &= \sum_{\alpha < \beta} Z_{\gt}^{\alpha,\beta}(\eta)= \sum_{\sigma:\sigma_U=\eta,\sum_{w\in W^+} \sigma_{w}< \sum_{w\in W^-} \sigma_{w}}\lambda^{|\sigma|}
\end{align*}
and $Z^\pm_{\gt}=\sum_\eta Z^\pm_{\gt}(\eta)$.
\begin{lemma}\label{l:gtExpectedPartition}
For every $d\geq 3$ there exists  constants $\theta^*(\lambda,d),\psi^*(\lambda,d)>0$ such that when $\lambda_c(d) < \lambda$ and $0<\theta(\lambda,d)<\theta^*(\lambda,d),0<\psi(\lambda,d)<\psi^*(\lambda,d)$ then the expected partition functions satisfy
\begin{equation}\label{e:gtEZa}
\sup_\eta \left| \frac{\E Z^\pm_{\gt}(\eta)}{\E Z^\pm_{\gt}} - Q_U^\pm(\eta) \right| = o(1)
\end{equation}
and
\begin{equation}\label{e:gtEZb}
\E Z^+_{\gt}=(1+o(1)) \E Z^-_{\gt}.
\end{equation}
\end{lemma}
\begin{proof}
Recall from Lemma~\ref{l:ratio1Moment} that
\begin{equation}\label{e:gtEZ1}
\E Z^{\alpha,\beta}_{\gt}(\eta)= (1+O(n^{-1/2})) C^* \left(\lambda \left(\frac{1-\alpha-\beta}{1-\beta} \right)^{d-1}\right)^{\eta^-}
\left(\lambda \left(\frac{1-\alpha-\beta}{1-\alpha} \right)^{d-1}\right)^{\eta^+} \E Z^{\alpha,\beta}_{\mathrm{MWW}}
\end{equation}
and that by \cite{MWW:09}*{Proposition 3.1},
\[
\E Z^{\alpha,\beta}_{\mathrm{MWW}} \approx \exp(\Phi_1(\alpha,\beta) n)
\]
where the approximation holds up to a polynomial factor in $n$.  In the proof of \cite{DFJ:02}*{Claim 2.2} it is shown that for fixed $\alpha$ (resp. $\beta$) $\Phi_1$ is maximized by setting $\beta=h(\alpha)$ (resp. $\alpha=h(\beta)$) where $h(x)=(1-x)[1-(x /(\lambda(1-x)))^{1/d}]$ was defined in the Section~\ref{s:preliminaries}.  Recall that in $\{(\alpha,\beta)\in \mathcal{T}: \alpha \geq \beta \}$ the function $\Phi_1$ is maximized at $(p^+,p^-)$.  Clearly we have that the functions $\Phi_1(\alpha,h(\alpha))$ (resp. $\Phi_1(h(\beta),\beta)$) are analytic in $\alpha$ (resp. $\beta$) when in a neighbourhood of $p^+$ (resp. $p^-$).  It follows by expanding as a Taylor series and noting that $(p^+,p^-)$ is a local maxima that for some integer $\ell \geq 2$ and constants $C,\epsilon>0 $ we have that
\[
\left|\Phi_1(\alpha,h(\alpha)) - \Phi_1(p^+,p^-)\right|\geq C|\alpha - p^+|^\ell
\]
when $|\alpha - p^+|\leq \epsilon$ and similarly for $\Phi_1(h(\beta),\beta)$.  This of course implies that when $\|(\alpha,\beta)-(p^+,p^-)\|_\infty \leq \epsilon$ then
\[
\left|\Phi_1(\alpha,\beta) - \Phi_1(p^+,p^-)\right|\geq C\|(\alpha,\beta) - (p^+,p^-)\|_\infty^\ell.
\]
Hence it follows that for large $n$,
\begin{equation}\label{e:gtEZ2}
\sum_{\alpha\geq \beta,\|(\alpha,\beta) - (p^+,p^-)\|_\infty > n^{-\frac1{2\ell}}} Z_{\gt}^{\alpha,\beta}(\eta) \leq \exp\left(-\frac{C}{2}n^{1/2}\right)\E Z^{+}_{\gt}(\eta).
\end{equation}
Setting $\theta^*(\lambda,d)=\psi^*(\lambda,d)=\frac1{5\ell}$ we have that $|U|\leq n^{\frac2{5\ell}}$ and hence
\begin{align}\label{e:gtEZ3}
\left(\lambda \left(\frac{1-\alpha-\beta}{1-\alpha} \right)^{d-1}\right)^{\eta^+}&= (1+o(1))\left(\lambda \left(\frac{1-p^+ -p^-}{1-p^+} \right)^{d-1}\right)^{\eta^+},\nonumber\\
\left(\lambda \left(\frac{1-\alpha-\beta}{1-\beta} \right)^{d-1}\right)^{\eta^-}&= (1+o(1))\left(\lambda \left(\frac{1-p^+ -p^-}{1-p^-} \right)^{d-1}\right)^{\eta^-}.
\end{align}
Combining equations~\eqref{e:gtEZ1}, \eqref{e:gtEZ2} and~\eqref{e:gtEZ3} we establish that
\begin{align}\label{e:gtEZ4}
\E Z^+_{\gt}(\eta) &= (1+o(1)) C^*\left(\lambda \left(\frac{1-p^+ -p^-}{1-p^-} \right)^{d-1}\right)^{\eta^-} \left(\lambda \left(\frac{1-p^+ -p^-}{1-p^+} \right)^{d-1}\right)^{\eta^+} \E Z^{+}_{\mathrm{MWW}}\nonumber\\
&=(1+o(1))  C^* \left(\frac{q^-}{1-q^-}\right)^{\eta^-} \left(\frac{q^+}{1-q^+} \right)^{\eta^+} \E Z^{+}_{\mathrm{MWW}}\nonumber\\
&=\frac{(1+o(1))C^*}{(1-q^+)^{m'}(1-q^-)^{m'}}  Q^+_U(\eta) \E Z^{+}_{\mathrm{MWW}}
\end{align}
where $\E Z^{+}_{\mathrm{MWW}}$ denotes $\sum_{\alpha\geq \beta} \E Z^{\alpha,\beta}_{\mathrm{MWW}}$, where the second line follows by equation~\eqref{e:pqRelation} and the final equality holds by the definition of $Q^+_U$.  Hence we have that
\[
\E Z^+_{\gt}=\frac{(1+o(1))C^*}{(1-q^+)^{m'}(1-q^-)^{m'}}\E Z^{+}_{\mathrm{MWW}}
\]
and that
\[
\sup_\eta \left| \frac{\E Z^+_{\gt}(\eta)}{\E Z^\pm_{\gt}} - Q_U^\pm(\eta) \right| = o(1)
\]
The analogous bound holds for $Z^-_{\gt}(\eta)$ which establishes equation~\eqref{e:gtEZa}.

Note that because of the slight asymmetry in the definition of $Z^+_{\gt}$ and $Z^-_{\gt}$ equation~\eqref{e:gtEZb} is not immediate by symmetry.  It follows from the fact that when $\lambda>\lambda_c$  by symmetry we have that
\[
\sum_{\alpha} \E Z^{\alpha,\alpha}_{\gt} =  \E Z^+_{\gt} - \E Z^-_{\gt}
\]
and
\[
\sum_{\alpha} \E Z^{\alpha,\alpha}_{\gt} \leq \exp(-\Omega(n)) \E Z^+_{\gt}
\]
since the maxima of $\Phi_1(\alpha,\beta)$ is not achieved with $\alpha=\beta$.  This completes the lemma.
\end{proof}

\subsection{Second Moment Analysis}
We now proceed to analyze the second moment of the partition function.
In~\cites{MWW:09} they showed that the second moment is given by
\begin{align}\label{e:MWW2Moment}
\E \left[Z_{\mathrm{MWW}}^{\alpha,\beta}(\eta)\right]^2 &= \lambda^{2(\alpha+\beta)} \binom{n}{\alpha n} \binom{n}{\beta n}
\sum_{\gamma,\delta} \binom{\alpha n}{\gamma n} \binom{(1-\alpha)n}{(\alpha-\gamma) n} \binom{\beta n}{\delta n} \binom{(1-\beta)n}{(\beta-\delta) n}\nonumber\\
&\quad \cdot  \left[ \frac{\binom{(1-2\beta+\delta)n}{\gamma n}}{\binom{n}{\gamma n}}
\sum_{\epsilon} \frac{\binom{(1-2\beta+\delta-\gamma)n}{\epsilon n}\binom{(\beta-\delta)n}{(\alpha-\gamma-\epsilon)n}}{\binom{(1-\gamma)n}{(\alpha-\gamma)n}}
\frac{\binom{(1-\beta-\gamma-\epsilon)n}{(\alpha-\gamma)n}}{\binom{(1-\alpha)n}{(\alpha-\gamma)n}}
\right]^d
\end{align}
where the sums run over $\gamma,\delta,\epsilon$ such that $\gamma n,\delta n,\epsilon n$ and equation~\eqref{e:greekRegion} holds.
Equation~\eqref{e:MWW2Moment} should be interpreted as follows:  The first line represents the number of ways of choosing a pair of configurations, both with size $\alpha$ on the  plus side and $\beta$ on the minus side with overlaps of $\gamma$ on the plus side and $\delta$ on the minus side.  The second line gives the probability that the pair of configurations are both independent sets in the random graph (see~\cite{MWW:09} for the interpretation of the sum).  Here the role of Condition~\ref{cond:technical} comes into play.  A simple approximation gives that
\begin{align*}
\exp(nf(\alpha,\beta,\gamma,\delta,\epsilon)) &\approx \lambda^{2(\alpha+\beta)} \binom{n}{\alpha n} \binom{n}{\beta n}
\binom{\alpha n}{\gamma n} \binom{(1-\alpha)n}{(\alpha-\gamma) n} \binom{\beta n}{\delta n} \binom{(1-\beta)n}{(\beta-\delta) n}\nonumber\\
&\quad \cdot  \left[ \frac{\binom{(1-2\beta+\delta)n}{\gamma n}}{\binom{n}{\gamma n}}
\frac{\binom{(1-2\beta+\delta-\gamma)n}{\epsilon n}\binom{(\beta-\delta)n}{(\alpha-\gamma-\epsilon)n}}{\binom{(1-\gamma)n}{(\alpha-\gamma)n}}
\frac{\binom{(1-\beta-\gamma-\epsilon)n}{(\alpha-\gamma)n}}{\binom{(1-\alpha)n}{(\alpha-\gamma)n}}
\right]^d
\end{align*}
up to polynomial terms in $n$.  As such the maximum of $f$ plays a crucial role in the second moment analysis.
The following result is by~\cite{MWW:09}*{Lemma 3.3}. While they only stated their result for $(\alpha,\beta)$ close to $(1/d,1/d)$ it is easy to verify that their proof holds in a neighbourhood of $(p^-,p^+)$ whenever Condition~\ref{cond:technical} holds.

\begin{lemma}[\cite{MWW:09}*{Lemma 3.3}]\label{l:MWWvarRatio}
For each $d\geq3$ suppose that Condition~\ref{cond:technical} holds.  Then there exists some $\chi>0$ such that when $|\alpha-p^-|,|\beta-p^+| < \chi$ we have that,
\[
\frac{\E \left(Z^{\alpha,\beta}_{\mathrm{MWW}}\right)^2}{\left( \E Z^{\alpha,\beta}_{\mathrm{MWW}}\right)^2} \to \tau^{\alpha,\beta}
\]
where
\[
\tau^{\alpha,\beta}=\frac{(1-\alpha-\beta-\alpha\beta)^d}{[(1-\alpha-\beta+2\alpha\beta)(1-\alpha-\beta)]^{\frac{d-1}{2}}
[(1-\alpha-\beta+d\alpha\beta)(1-\alpha-\beta-(d-2)\alpha\beta)]^{\tfrac12}}.
\]
\end{lemma}
Next we show  the analogous result for $\gt$ conditioned on $\sigma_U$ by estimating the ratio of the second moments of the partition functions of the graphs.
\begin{lemma}\label{l:gt2Moment}
For each $d\geq3$ suppose that Condition~\ref{cond:technical} holds.  Then there exists some $\chi>0$ such that when $|\alpha-p^-|,|\beta-p^+| < \chi$ we have that for all $\eta\in\{0,1\}^U$,
\begin{equation}\label{e:ratio2Moment}
\frac{\E \left(Z^{\alpha,\beta}_{\gt}(\eta)\right)^2}{\E \left(Z^{\alpha,\beta}_{\mathrm{MWW}}\right)^2} =  (1 + o(1))(C^*)^2\left(\lambda \left(\frac{1-\alpha-\beta}{1-\beta} \right)^{d-1}\right)^{2\eta^-}
\left(\lambda \left(\frac{1-\alpha-\beta}{1-\alpha} \right)^{d-1}\right)^{2\eta^+},
\end{equation}
and hence
\[
\frac{\E\left(Z^{\alpha,\beta}_{\gt}(\eta)\right)^2}{\left(\E Z^{\alpha,\beta}_{\gt}(\eta) \right)^2} \to \tau^{\alpha,\beta}.
\]
\end{lemma}

\begin{proof}
Repeating the analysis of~\cite{MWW:09} the analogous
\begin{align}\label{e:gt2Moment}
\E \left(Z^{\alpha,\beta}_{\gt}(\eta)\right)^2 = \lambda^{2(\alpha+\beta)} \binom{n}{\alpha n} \binom{n}{\beta n}
\sum_{\gamma,\delta} \binom{\alpha n}{\gamma n} \binom{(1-\alpha)n}{(\alpha-\gamma) n} \binom{\beta n}{\delta n} \binom{(1-\beta)n}{(\beta-\delta) n}\nonumber\\
 \cdot  \left[ \frac{\binom{(1-2\beta+\delta)n+m-\eta^-}{\gamma n+\eta^+}}{\binom{n+m}{\gamma n+\eta^+}}
\sum_{\epsilon} \frac{\binom{(1-2\beta+\delta-\gamma)n+m - \eta^+ -\eta^-}{\epsilon n} \binom{(\beta-\delta)n}{(\alpha-\gamma-\epsilon)n}}{\binom{(1-\gamma)n +m - \eta^+ -\eta^-}{(\alpha-\gamma)n}}
\frac{\binom{(1-\beta-\gamma-\epsilon)n+m-\eta^+}{(\alpha-\gamma)n}}{\binom{(1-\alpha)n+m - \eta^-}{(\alpha-\gamma)n}}
\right]^d
\end{align}
By Lemma~\ref{l:MWWfMax} the unique maxima of $g_{\alpha,\beta}(\gamma,\delta,\epsilon)$ is $(\gamma^*,\delta^*,\epsilon^*)$ and it was shown in~\cite{MWW:09} that $g_{\alpha,\beta}$ decays quadratically from this point.  Consequently, with
\[
\mathcal{A}=\{(\gamma,\delta,\epsilon):|\gamma-\gamma^*|,|\delta-\delta^*|,|\epsilon-\epsilon^*|\leq n^{-1/4}\}
\]
the contribution from terms with $(\gamma,\delta,\epsilon)\not\in \mathcal{A}$ is  $\exp(-\Omega(n^{1/2}))$ and so can be omitted.  Setting
\begin{align*}
\kappa^{\alpha,\beta}_{\gt}(\eta)=\lambda^{2(\alpha+\beta)} \frac{\binom{(1-2\beta+\delta)n+m-\eta^-}{\gamma n+\eta^+}}{\binom{n+m}{\gamma n+\eta^+}}
\frac{\binom{(1-2\beta+\delta-\gamma)n+m - \eta^+ -\eta^-}{\epsilon n} \binom{(\beta-\delta)n}{(\alpha-\gamma-\epsilon)n}}{\binom{(1-\gamma)n +m - \eta^+ -\eta^-}{(\alpha-\gamma)n}}
\frac{\binom{(1-\beta-\gamma-\epsilon)n+m-\eta^+}{(\alpha-\gamma)n}}{\binom{(1-\alpha)n+m - \eta^-}{(\alpha-\gamma)n}}\\
\kappa^{\alpha,\beta}_{\mathrm{MWW}}=\frac{\binom{(1-2\beta+\delta)n}{\gamma n}}{\binom{n}{\gamma n}}
\frac{\binom{(1-2\beta+\delta-\gamma)n}{\epsilon n}\binom{(\beta-\delta)n}{(\alpha-\gamma-\epsilon)n}}{\binom{(1-\gamma)n}{(\alpha-\gamma)n}}
\frac{\binom{(1-\beta-\gamma-\epsilon)n}{(\alpha-\gamma)n}}{\binom{(1-\alpha)n}{(\alpha-\gamma)n}}
\end{align*}
and recalling that $\gamma^*=\alpha^2,\delta^*=\beta^2,\epsilon^*=\alpha(1-\alpha-\beta)$ we have that
for $(\gamma,\delta,\epsilon)\in \mathcal{A}$
\begin{align}\label{e:kappaComparison}
\frac{\kappa^{\alpha,\beta}_{\gt}(\eta)}{\kappa^{\alpha,\beta}_{\mathrm{MWW}}} &= (1+o(1)) \lambda^{2(\eta^+ + \eta^-)} \left(\frac{1-2\beta+\delta}{1-2\beta+\delta-\gamma} \right)^{m-\eta^-}\nonumber\\
& \quad \cdot
\left(\frac{1-2\beta+\delta-\gamma}{\gamma} \right)^{\eta^+}
\left(1-\gamma\right)^{m} \left(\frac{\gamma}{1-\gamma} \right)^{\eta^+}
\left(\frac{1-2\beta+\delta-\gamma}{1-2\beta+\delta-\gamma-\epsilon} \right)^{m-\eta^+\eta^-}\nonumber\\
&\quad \cdot \left(\frac{1-\beta-\gamma-\epsilon}{1-\beta-\alpha-\epsilon} \right)^{m-\eta^+-\eta^-}
\left(\frac{1-\alpha}{1-\gamma} \right)^{m-\eta^+}
\left(\frac{1-2\alpha+\gamma}{1-\alpha} \right)^{m-\eta^+}\nonumber\\
&=(1+o(1)) (C^*)^2 \left(\lambda \left(\frac{1-\alpha-\beta}{1-\beta} \right)^{d-1}\right)^{2\eta^-}
\left(\lambda \left(\frac{1-\alpha-\beta}{1-\alpha} \right)^{d-1}\right)^{2\eta^+}
\end{align}
where the first line follows by equation~\eqref{e:bionmialPerturb} and the second follows by approximating $(\gamma,\delta,\epsilon)$ with $(\gamma^*,\delta^*,\epsilon^*)$ and simplifying.  Now comparing equations~\eqref{e:MWW2Moment} and~\eqref{e:gt2Moment} (noting here that we can neglect terms with $(\gamma,\delta,\epsilon)\not\in \mathcal{A}$) we have that
\[
\frac{\E \left(Z^{\alpha,\beta}_{\gt}(\eta)\right)^2}{\E \left(Z^{\alpha,\beta}_{\mathrm{MWW}}\right)^2} = (1+o(1)) (C^*)^2 \left(\lambda \left(\frac{1-\alpha-\beta}{1-\beta} \right)^{d-1}\right)^{2\eta^-}
\left(\lambda \left(\frac{1-\alpha-\beta}{1-\alpha} \right)^{d-1}\right)^{2\eta^+},
\]
which establishes equation~\eqref{e:ratio2Moment}.  Equation~\eqref{e:gt2Moment} then follows from Lemmas~\ref{l:ratio1Moment} and~\ref{l:MWWvarRatio}.
\end{proof}

\subsection{Small Graph Conditioning Method}

As we showed in the previous section the ratio of second moment and the first moment squared of $Z^{\alpha,\beta}_{\gt}(\eta)$ converges to a constant $\tau^{\alpha,\beta}$.  Unfortunately, since $\tau^{\alpha,\beta}>1$ we can not directly apply the second moment method to get high probability bounds. Instead we follow the approach of~\cite{MWW:09} and use the small graph conditioning method.  Our proofs differ very minimally from theirs and as such we only comment on the necessary changes.  At a high level, the method says that small cycles in the graph ``explain'' the variance which provides good lower bounds on $Z^{\alpha,\beta}_{\gt}(\eta)$.  The following is taken from~\cite{MWW:09} which itself is presented as a special case of a  results of \cite{Wormald:99} and \cite{JLR:00}.
\begin{theorem}[\cite{MWW:09}*{Theorem 7.1}]\label{t:generalSmallGraph}
Let $\lambda_i>0$ and $\delta_i >0$ be real numbers for $i = 1, 2, \ldots$ Let $\omega(n) \to 0$ and suppose that for
each $n$ there are random variables $X_i = X_i(n), i = 1, 2,\ldots$ and $Y = Y (n)$, all defined on the same probability
space $G = G_n$ such that $X_i$ are nonnegative integer valued, Y is nonnegative and $EY > 0$ (for n sufficiently large).
Suppose furthermore that
\begin{enumerate}
\item For each $k \geq 1$, the variables $X_1, \ldots ,X_k$ are asymptotically independent Poisson random variables with $\E X_i\to\lambda_i$
\item For every finite sequence $m_1, \ldots ,m_k$ of nonnegative integers,
\begin{equation}\label{e:cycleRatio}
\frac{\E\left(Y\prod_{i=1}^k[X_i]_{m_k}\right)}{\E Y}\to \prod_{i=1}^k (\lambda_i (1+\delta_i))^{m_i}
\end{equation}
where $[X]_m=\prod_{i=0}^{m-1} (X-i)$, denotes the falling factorial.
\item That $\sum \lambda_i \delta_i^2 < \infty$,
\item That $\E Y^2/(\E Y)^2 \leq \exp(\sum \lambda_i\delta_i^2)+o(1)$ as $n\to\infty$.
\end{enumerate}
Then $\P(Y > \omega(n) \E Y)\to 1$.
\end{theorem}
We set $Y=\lambda^{-(\alpha+\beta)n} Z^{\alpha,\beta}_{\gt}(\eta)$ and let $X_i$ be the number of cycles of length $i$ whose vertices lie in $W$ (which is of course 0 when $i$ is odd).  The following lemma has an essentially identical proof to the proof of~\cite{MWW:09}*{Lemma 7.3} and follows from standard methods.
\begin{lemma}\label{l:cycleCount}
For even $i$ the number of cycles  are asymptotically Poisson with means $\lambda_i=r(d,i)/i$ where $r(d,i)$ counts the number of proper $d$-colourings of a cycle of size $i$.
\end{lemma}
Next, the main step is to determine the limit of equation~\eqref{e:cycleRatio}.  The proof of the following lemma follows with only very minor modifications form that of \cite{MWW:09}*{Lemma 7.4 and 7.5}.
\begin{lemma}\label{l:smallGraphCovariance}
For all $(\alpha,\beta)$ in the interior of $\mathcal{T}$ and all $\eta\in\{0,1\}^U$ we have that:
\begin{equation*}
\frac{\E\left(Y\prod_{i=1}^k[X_i]_{m_k}\right)}{\E Y}\to \prod_{i=1}^k (\lambda_i (1+\delta_i))^{m_i}
\end{equation*}
where $\lambda_i=r(d,i)/i$ and
\[
\delta_i=\left(\frac{\alpha\beta}{(1-\alpha)(1-\beta)}\right)^{i/2}
\]
for even $i$.
\end{lemma}

\begin{proof}
We follow as far as possible the proof of \cite{MWW:09}*{Lemma 7.4}.  We consider just the case where a single $m_i=1$ and the others are all as the extension to general $m_i$'s is exactly as in \cite{MWW:09}*{Lemma 7.5}.  We follow the notation is from \cite{MWW:09} with only slight modifications to our setting
\begin{itemize}
\item Let $\Upsilon\in \{0,1 \}^W$  with $\sum_{w\in W^+} \Upsilon_w=\alpha n$ and with $\sum_{w\in W^-} \Upsilon_w=\beta n$.
\item Denote by $\xi$ a proper $d$-edge-coloured rooted, oriented $i$-cycle (from amongst the $r(d, i)$ possibilities), in which the vertices are
2-coloured, black and white, with no two black vertices adjacent. The color of the edges will prescribe
which of the $d$ perfect matchings an edge of a (potential) cycle will belong to. The black vertices will
prescribe which of the cycle vertices are members of $\{w\in W:\Upsilon_w = 1\}$.
\item We let $\zeta$ denotes a position that an $i$-cycle can be in (i.e. the exact vertices of $W$ it traverses, in order) such that the prescription
of the vertex colors of $\xi$ is satisfied.  (Note this was denoted as $\eta$ in~\cite{MWW:09}).
\item Denote by $P1$ is the probability that a random graph $\gt$ contains a cycle $C$ in the given position $\zeta$ with the edge colors
prescribed by $\xi$ in accordance with which matchings contain the edges of $C$.
\item We denote by $P2$  the conditional probability that in the random graph $\gt$ that the set $\{w\in W:\Upsilon_w = 1\}\cup \{u\in U:\eta_u = 1\}$ is an independent set, given that it
contains $C$ as in the definition of $P1$.
\item Denote by $P3$ the probability that in the random graph $\gt$ the set $\{w\in W:\Upsilon_w = 1\}\cup \{u\in U:\eta_u = 1\}$ is an independent set.
\end{itemize}
Analogously to equation (18) of \cite{MWW:09} we have that
\begin{equation}
\frac{\E Y X_i}{\E Y} = \frac1i \sum_{\xi} \sum_{\zeta} \frac{P1 \ P2}{P3}
\end{equation}
as the probabilities are independent of $\Upsilon$.  It is immediate from the definition that $P1=(1+o(1))n^{-i}$.

Now closely following the notation of \cite{MWW:09}  for $k=1,\ldots,d$ let $e(k)$ denote the number of edges of colour $k$ in $\xi$. Denote by $f_\pm(k)$ the number of black vertices adjacent to edges of colour $k$ in the sets $\{w\in W^\pm:\Upsilon_w = 1\}$.  Assuming that $\xi$ is compatible with $\{w\in W^\pm:\Upsilon_w = 1\}$ then the probability that the remaining edges also respect the independents sets  is given by,
\begin{align}\label{e:P2expression}
P2 &= \frac{\left(\prod_{k=1}^{d-1} \binom{n+m'-\beta n -\eta^- -e(k) + f_-(k)}{\alpha n + \eta^+ - f_+(k)}\right) \binom{n-\beta n-e(d) + f_-(d)}{\alpha n - f_+(d)} }
{\left(\prod_{k=1}^{d-1} \binom{n+m' -e(k) }{\alpha n + \eta^+ - f_+(k)}\right) \binom{n -e(d) }{\alpha n - f_+(d)}}.
\end{align}
Now by Lemma~\ref{l:bionmialPerturb} we have that
\begin{align}\label{e:P2componentA}
\frac{\binom{n+m'-\beta n -\eta^- -e(k) + f_-(k)}{\alpha n + \eta^+ - f_+(k)}}{\binom{n+m'-\beta n -\eta^- }{\alpha n + \eta^+ }}
&= (1+o(1))\left(\frac{n+m'-\beta n -\eta^-}{n+m'-\beta n -\eta^- - \alpha n - \eta^+}\right)^{-e(k) + f_-(k)}\nonumber\\
&\qquad \cdot \left(\frac{n+m'-\beta n -\eta^- - \alpha n - \eta^+}{\alpha n + \eta^+}\right)^{- f_+(k)}\nonumber\\
&= (1+o(1)) \left(\frac{1-\beta }{1-\beta -\alpha }\right)^{-e(k) + f_-(k)}
\left(\frac{1-\beta-\alpha}{\alpha}\right)^{- f_+(k)}
\end{align}
where we used the fact that $m,\eta^+,\eta^-=O(n^{-1/4})$ and $e(k), f_-(k),f_+(k)=O(1)$.  Similarly we have that
\begin{align}\label{e:P2componentB}
\frac{\binom{n-\beta n-e(d) + f_-(d)}{\alpha n - f_+(d)}}{\binom{n-\beta n}{\alpha n}}
&= (1+o(1))  \left(\frac{1-\beta }{1-\beta -\alpha }\right)^{-e(d) + f_-(d)}
\left(\frac{1-\beta-\alpha}{\alpha}\right)^{- f_+(d)}\nonumber\\
\frac{\binom{n+m' -e(k) }{\alpha n + \eta^+ - f_+(k)}}{\binom{n+m'  }{\alpha n + \eta^+ }}
&= (1+o(1)) \left(\frac1{1-\alpha}\right)^{-e(k)}\left( \frac{1-\alpha}{\alpha} \right)^{-f_+(k)}\nonumber\\
\frac{\binom{n -e(d) }{\alpha n + - f_+(d)}}{\binom{n  }{\alpha n  }}
&= (1+o(1)) \left(\frac1{1-\alpha}\right)^{-e(d)}\left( \frac{1-\alpha}{\alpha} \right)^{-f_+(d)}
\end{align}
By equation~\eqref{e:gt1moment} we have that
\begin{align}\label{e:P3expression}
P3=\left[ \frac{\binom{n+m'-\beta n -\eta^-}{\alpha n + \eta^+}}{\binom{n+m'}{\alpha n+\eta^+}} \right]^{d-1}
\frac{\binom{n-\beta n}{\alpha n}}{\binom{n}{\alpha n}}
\end{align}
Now let $j_\pm(\xi)=\frac12\sum_{k=1}^d f_\pm(k)$ denote the number of black vertices in $V^\pm$ according to $\xi$ and recall that $i=\sum_{k=1}^d e(k)$.  Combining equations \eqref{e:P2expression},\eqref{e:P2componentA},\eqref{e:P2componentB} and~\eqref{e:P3expression} we have that
\[
\frac{P2}{P3} = (1+o(1))\frac{(1-\alpha-\beta)^{i-2j_- -2_+}}{(1-\alpha)^{i-2j_+}(1-\beta)^{i-2j_-}}.
\]
Now letting $Pk^{\mathrm{MWW}}$ denote the corresponding probabilities in Lemma~7.4 of \cite{MWW:09} we note that
\[
P1=(1+o(1)) P1^{\mathrm{MWW}}, \qquad \frac{P2}{P3} = (1+o(1)) \frac{P2^{\mathrm{MWW}}}{P3^{\mathrm{MWW}}}
\]
Hence
\begin{align*}
\frac{\E Y X_i}{\E Y} &= (1+o(1))\frac1i \sum_{\xi} \sum_{\zeta} \frac{P1 \ P2}{P3}\\
&= (1+o(1)) \frac1i \sum_{\xi} \sum_{\zeta} \frac{P1^{\mathrm{MWW}} \ P2^{\mathrm{MWW}}}{P3^{\mathrm{MWW}}}\\
&= (1+o(1)) \lambda_i (1+\delta_i),
\end{align*}
where the final term is the main result of \cite{MWW:09}*{Lemma 7.4}.  The complete result for general $m_i$'s follows similarly to \cite{MWW:09}*{Lemma 7.5} which completes the lemma.
\end{proof}

\begin{lemma}\label{l:whpSmallGraph}
If $d\geq 3$ and $(\alpha,\beta)$ is in the interior of $\mathcal{T}$ and the function $g_{\alpha,\beta}$ achieves its unique maxima in~\eqref{e:greekRegion} at $(\alpha^2,\beta^2,\alpha(1-\alpha-\beta))$ then for all $\eta \in \{0,1\}^U$,
\begin{equation}\label{e:whpBound}
\sup_\eta \P\left(Z^{\alpha,\beta}_{\gt}(\eta) < \frac2{\sqrt{n}} \E Z^{\alpha,\beta}_{\gt}(\eta) \right) \to 0.
\end{equation}
\end{lemma}
\begin{proof}
The result follows from an application of Theorem~\ref{t:generalSmallGraph}, taking the $i$ to be even, $\lambda_i=r(d,i)/i$ and
$\delta_i=\left(\frac{\alpha\beta}{(1-\alpha)(1-\beta)}\right)^{i/2}$.  Condition (1) of the theorem holds by Lemma~\ref{l:cycleCount}.  Condition (2) holds by Lemma~\ref{l:smallGraphCovariance}.  Conditions (3) and (4) hold as a consequence of Lemma~\ref{l:gt2Moment} and \cite{MWW:09}*{Lemma 7.6}.  Taking $\omega(n)=\frac2{\sqrt{n}}$, the result follows.
\end{proof}

In Lemma~\ref{l:gtExpectedPartition} we gave estimates of the expected conditional partition functions.  In this subsection we of the small graph conditioning method results and give with high probability type estimates for the conditional partition functions.
\begin{theorem}\label{t:gtWHP}
For every $d\geq 3$ and $\lambda>\lambda_c$ such that Condition~\ref{cond:technical} holds there exists a positive constant $\varepsilon(d)>0$ and constants $\theta^*(\lambda,d),\psi^*(\lambda,d)>0$ such that the partition functions satisfy the following asymptotic almost sure statements,
\begin{equation}\label{e:gtWHPa}
\sup_{\eta \in \{0,1\}^U} \P\left(Z^\pm_{\gt}(\eta) < \frac1{\sqrt{n}} \E Z^{\pm}_{\gt}(\eta) \right) \to 0.
\end{equation}
\end{theorem}
\begin{proof}
Condition~\ref{cond:technical} guarantees that for $(\alpha,\beta)$ in a neighborhood of $(p^-,p^+)$ that $g_{\alpha,\beta}$ achieves its unique maxima in~\eqref{e:greekRegion} at $(\alpha^2,\beta^2,\alpha(1-\alpha-\beta))$.  For sufficiently small $\delta>0$ then by Lemma~\ref{l:whpSmallGraph} we have that
\[
\sup_\eta \sup_{(\alpha,\beta)\in\mathcal{S}} \P\left(Z^{\alpha,\beta}_{\gt}(\eta) < \frac2{\sqrt{n}} \E Z^{\alpha,\beta}_{\gt}(\eta) \right) \to 0
\]
where $\mathcal{S}=\{(\alpha,\beta):\|(\alpha,\beta)-(p^+,p^-)\|_\infty<\delta\}$ and hence
\[
\sup_\eta  \P\left( \sum_{(\alpha,\beta)\in\mathcal{S}} Z^{\alpha,\beta}_{\gt}(\eta) < \frac3{2\sqrt{n}}
\sum_{(\alpha,\beta)\in\mathcal{S}} \E Z^{\alpha,\beta}_{\gt}(\eta) \right) \to 0.
\]
By equation~\eqref{e:gtEZ2} we have that for all $\eta\in\{0,1\}^U$,
\[
\sum_{\alpha\geq \beta,(\alpha,\beta)\not\in\mathcal{S}} Z_{\gt}^{\alpha,\beta}(\eta) \leq \exp\left(-\frac{C}{2}n^{1/2}\right)\E Z^{+}_{\gt}(\eta).
\]
and hence we have that
\[
\sup_\eta \P\left(Z^+_{\gt}(\eta) < \frac1{\sqrt{n}} \E Z^{+}_{\gt}(\eta) \right) \to 0.
\]
The analogous statement for $Z^-_{\gt}(\eta)$ holds similarly which completes the lemma.
\end{proof}

\section{Reconstruction on the tree}\label{s:reconstruction}
Our proof now takes a detour through the reconstruction problem on the tree.  This problem concerns determining which Gibbs measures on the tree are extremal, or equivalently when the tail $\sigma$-algebra is trivial or when point-to-set correlations converge to 0 in the distance of the point to the set \cite{MosPer:03}.  In our setting the measures $\hmu_\pm$ are extremal so we automatically have that non-reconstruction holds.  We will use facts about the rate of decay of point-to-set correlations to establish that $\sigma_V$ is essentially independent of $\sigma_U$ conditioned on the phase.
In most cases the reconstruction problem has been considered in the case of the translation invariant free measure (see~\cite{BST:10} for recent progress on the hardcore model) but we will be interested in the case of the semi-translation invariant measures $\hat{\mu}_\pm$ on $\dtree$  and as such results from the literature do not directly apply here.

The reconstruction problem has for the most part been studied in the case of Markov models on trees with a single transition kernel $M$.
In this theory the key role is played by the $\lambda_*$ the second eigenvalue of the transition matrix.  The famous Kesten-Stigum bound \cite{KesSti:66,MosPer:03} states that there is reconstruction when $\lambda_*^2(d-1) > 1$ while results of~\cite{JanMos:04} show that if non-reconstruction holds and $\lambda_*^2(d-1) < 1$ then point to set correlations decay exponentially quickly.  In our setting, however, the Gibbs measure is semi-translation invariant and the Markov model is given by a pair of alternating Markov transition kernels, $M^\pm$ defined below.

With minor modifications the proof of \cite{BCMR:06} (or also \cite{Sly:09} or \cite{JanMos:04}) can be adapted to the semi-translation invariant setting.  Here the role of $\lambda_*$ is played by the second eigenvalue of $M_1 M_2$ and there is reconstruction when $\lambda^2_* (d-1)^2 > 1$ and exponential decay of correlations when there is non-reconstruction and $\lambda^2_* (d-1)^2 < 1$.  The term $(d-1)^2$ is explained by the fact that this this the branching from two levels of the tree.  Using the methods of~\cite{Sly:09} which build on the work of~\cite{BCMR:06} we establish the necessary decay of correlations result.

While we will be interested in the measure $\hmu$ it will be most convenient to work first on an adjusted Markov model $\tilde{\xi}^\pm$ on the tree $\dtree$ taking values in $\{0,1 \}^{\dtree}$ and then transfer results to $\hmu$.  The spin $\tilde{\xi}^\pm_\rho$ is chosen according to
\[
\P\left[ \tilde{\xi}^\pm_\rho = 1 \right] = 1 - \P\left[ \tilde{\xi}^\pm_\rho = 1 \right] = p^\pm.
\]
For the other vertices of the graph the values of $\tilde{\xi}^\pm$ will be propagated along the tree given though Markov transitions given by alternating transition kernels.  Specifically if vertex $u$ is the parent of $v$ in the tree then the spin at $v$ is defined according to the probabilities
\[
P(\tilde{\xi}^s_v = j|\sigma_u=i)=M^{s(-1)^{|v|}}_{i+1,j+1}.
\]
for $s\in\{-1,+1\}$ and $i,j\in{0,1}$ and where $|v|=d(\rho,v)$ and
\[
M^{+1}=\left(
      \begin{array}{cc}
        1-q^+ & q^+ \\
        1 & 0 \\
      \end{array}
    \right),
\qquad
M^{-1}=\left(
      \begin{array}{cc}
        1-q^- & q^- \\
        1 & 0 \\
      \end{array}
    \right).
\]
Viewing $\dtree\subset\mathbbm{T}_d$ we have that the measure of $\tilde{\xi}^\pm$ is simply the projection of $\mu^\pm$ to $\dtree$ (had we instead chosen $\P\left[ \tilde{\xi}^\pm_\rho = 1 \right]=q^\pm$ the $\tilde{\xi}^\pm$ would be given by $\hmu^\pm$).  It follows that
\[
\P\left[ \tilde{\xi}^s_v = 1 \right]=\begin{cases}
p^+ &\hbox{if }s(-1)^{|v|}=+1\\
p^- &\hbox{if }s(-1)^{|v|}=-1
\end{cases}.
\]
For a vertex $v\in T$ let $\dtree^v$ denote the subtree of descendants
of $v$ (including $v$).  Observe that the measure $\tilde{\xi}^s$ restricted to $\dtree^v$ is equal in distribution to $\tilde{\xi}^{s(-1)^{|v|}}$ on $\dtree$ appropriately shifted.  Now let $S_{v,\ell}$ denote the set of vertices in $\dtree$ which are $\ell$ levels below $v$ and let $\xi^\pm_{v,\ell}$ denote the configuration on $S_{v,\ell}$.  For a configuration $A$ on $S_{v,\ell}$ and $s\in\{-1,+1\}$ define the posterior function $\tilde{h}^s_{v,\ell}$ as
\[
\tilde{h}^s_{v,\ell}(A) = \P(\tilde{\xi}^s_v= 1 | \tilde{\xi}^s_{v,\ell} = A),
\]
We set
\[
\tilde{X}_{v,\ell,s} = \tilde{h}^s_{v,\ell}(\tilde{\xi}^s_{v,\ell})
\]
for $s\in\{+,-\}$.  Now since the measures $\mu^\pm$ are extremal it follows (see e.g. \cite{Mossel:04}) that
\begin{equation}\label{e:extremalConvergence}
\tilde{X}_{v,\ell,s} \stackrel{a.s}{\to} p^{s(-1)^{|v|}}.
\end{equation}
Moreover, if $u_1,\ldots,u_{d-1}$ are the children of $\rho$ then by standard tree recursions for Gibbs measures,
\begin{align}\label{e:treeRecursion}
\tilde{X}_{\rho,\ell,s} &= \frac{p^{s}\prod_{i=1}^d \frac1{1-p^{-s}}[1-\tilde{X}_{u_i,\ell-1,s}]}{p^{s}\prod_{i=1}^d (1-p^{-s})^{-1}[1-\tilde{X}_{u_i,\ell-1,s}]
+ (1-p^{s})\prod_{i=1}^d \left[\frac{q^{-s}}{p^{-s}}\tilde{X}_{u_i,\ell-1,s} + \frac{1-q^{-s}}{1-p^{-s}}[1-\tilde{X}_{u_i,\ell-1,s}]\right]}\nonumber\\
&= \frac{p^{s}\prod_{i=1}^d [1-\frac{\tilde{X}_{u_i,\ell-1,s}-p^{-s}}{1-p^{-s}}]}{p^{s}\prod_{i=1}^d [1-\frac{\tilde{X}_{u_i,\ell-1,s}-p^{-s}}{1-p^{-s}}]
+ (1-p^{s})\prod_{i=1}^d \left[1+[\tilde{X}_{u_i,\ell-1,s}-p^{-s}]\frac{p^s}{(1-p^{-s})(1-p^{s})}\right]} =:\frac{\mathcal{A}}{\mathcal{B}},
\end{align}
where the second inequality follows from equation~\eqref{e:pqRelation}.  Next, similarly to \cite{Sly:09}, we set
\[
x_{\ell,s} = \E^1 \tilde{X}_{\rho,\ell,s}-p^s.
\]
We will let $\E^1$ (resp. $\E^0$) denote the expectation conditional on the $\tilde{\xi}^s_\rho=1$ (resp. $\tilde{\xi}^s_\rho=0$).  With $u_1,\ldots,u_{d-1}$ the children of $\rho$ we have the following relationships of the $\tilde{X}_{u,\ell,s}$.
\begin{lemma}\label{l:reconProperties}
The following hold:
\begin{itemize}
\item Conditional on $\tilde{\xi}_\rho^s$ the $\tilde{X}_{u_i,\ell,s}$ are conditionally independent.
\item Also $\E^1(\tilde{X}_{u_i,\ell,s}-p^{-s}) = \frac{-p^{-s}}{1-p^{-s}} x_{\ell,-s}$.
\item\label{lab:varEquality} We have that $x_{\ell,s} = (p^s)^{-1} \E( \tilde{X}_{\rho,\ell,s}-p^{s} )^2$.
\item For all integers $k\geq 1$ we have that
\[\E^1( \tilde{X}_{u_i,\ell,s}-p^{-s} )^k = O(x_{\ell,-s}).\]
\end{itemize}
\end{lemma}
\begin{proof}
The first part follows from the Markov property of $\tilde{\xi}$.  The second follows from the fact that $\E \tilde{X}_{u_i,\ell-1,s}-p^{-s}=0$.  The third part follows from the proof of Lemma 2.2 of \cite{Sly:09}. Finally for the forth part we have that
\begin{align*}
\E^1( \tilde{X}_{u_i,\ell,s}-p^{-s} )^k &= \E\left[( \tilde{X}_{u_i,\ell,s}-p^{-s} )^k \mid \tilde{\xi}_{v_i}=0\right]\\
&= \E^0 \left[( \tilde{X}_{\rho,\ell,-s}-p^{-s} )^k \right].
\end{align*}
Now since $|\tilde{X}_{\rho,\ell,-s}-p^{-s}|\leq 1$ we have that
\[
\E^0 \left[( \tilde{X}_{\rho,\ell,-s}-p^{-s} )^k \right] \leq \E^0 \left[( \tilde{X}_{\rho,\ell,-s}-p^{-s} )^2 \right]
\]
for $k\geq 3$.  When $k=1$ we have
\[
\E^0 \left[( \tilde{X}_{\rho,\ell,-s}-p^{-s} ) \right] = -\frac{p^{-s}}{1-p^{-s}} \E^1 \left[( \tilde{X}_{\rho,\ell,-s}-p^{-s} ) \right] = O(x_{\ell,-s})
\]
while when $k=2$ we have that
\[
\E^0 \left[( \tilde{X}_{\rho,\ell,-s}-p^{-s} )^2 \right] \leq (1-p^{-s})^{-1}\E \left[( \tilde{X}_{\rho,\ell,-s}-p^{-s} )^2 \right] = O(x_{\ell,-s})
\]
which completes the proof.
\end{proof}
We now expand out equation~\eqref{e:treeRecursion} as
\begin{equation}\label{e:recursionExpansion}
\frac{\mathcal{A}}{\mathcal{B}} = \mathcal{A} - \mathcal{A}(\mathcal{B}-1) + (\mathcal{B}-1)^2\frac{\mathcal{A}}{\mathcal{B}} \leq \mathcal{A} - \mathcal{A}(\mathcal{B}-1) + (\mathcal{B}-1)^2
\end{equation}
since $\mathcal{A}\leq \mathcal{B}$.  We may expand out $\mathcal{B}-1$ and can express in the form
\[
\mathcal{B}-1 = \sum_{\alpha\in \{0,1\}^{d-1}} c_\alpha \prod_{i=1}^{d-1} (\tilde{X}_{u_i,\ell-1,s}-p^{-s})^{\alpha_i}
\]
where for some constants $c_\alpha$.  Moreover, $c_\alpha=0$ if $|\alpha|\in\{0,1\}$ where $|\alpha|=\sum_{i=1}^{d-1}\alpha_i$.  Since
\[
\E^1 \prod_{i=1}^{d-1} (\tilde{X}_{u_i,\ell-1,s}-p^{-s})^{\alpha_i} = O(x^{|\alpha|}_{\ell-1,-s})
\]
it follows from Lemma~\ref{l:reconProperties} that
\begin{equation}\label{e:calAExpansion}
\E^1 \left[ -\mathcal{A}(\mathcal{B}-1) + (\mathcal{B}-1)^2 \right] = O(x^2_{\ell-1,-s}).
\end{equation}
Similarly we have that
\begin{align}\label{e:calBExpansion}
\E^1 \mathcal{A} &= p^s + \sum_{i=1}^d \E^1 \frac{p^s}{1-p^{-s}} [\tilde{X}_{u_i,\ell-1,s}-p^{-s}] +  O(x^2_{\ell-1,-s}) \nonumber\\
&= p^s + \frac{(d-1)p^s p^{-s}}{(1-p^{-s})^2} x_{\ell-1,-s} +  O(x^2_{\ell-1,-s}).
\end{align}
Combining equations \eqref{e:recursionExpansion},\eqref{e:calAExpansion} and~\eqref{e:calBExpansion} we have that
\[
x_{\ell,s} = \E^1 \tilde{X}_{\rho,\ell,s} - p^s = \frac{(d-1)p^s p^{-s}}{(1-p^{-s})^2} x_{\ell-1,-s} +  O(x^2_{\ell-1,-s}).
\]
and after iterating we have that
\begin{align}\label{e:eRecursiveXn}
x_{\ell,s} = \E^1 \tilde{X}_{\rho,\ell,s} - p^s &= \frac{(d-1)^2 (p^s p^{-s})^2}{(1-p^s)^2(1-p^{-s})^2} x_{\ell-2,s} +  O(x^2_{\ell-2,s})\nonumber\\
 &= (d-1)^2 (q^+ q^{-})^2 x_{\ell-2,s} +  O(x^2_{\ell-2,s}) .
\end{align}
Now by equation~\eqref{e:extremalConvergence} we have that $x_{\ell,s} \to 0$ as $\ell \to 0$ and hence by equation~\eqref{e:eRecursiveXn} it converges exponentially fast to 0 as one of our initial assumptions in equation~\eqref{e:extraConditions} was that $q^+ q^- (d-1) < 1$.  It follows that by the second part of Lemma~\ref{l:reconProperties}  that there exist constants $\tilde{C}_1(\lambda,d),\tilde{C}_2(\lambda,d)>0$ such that,
\begin{equation}\label{e:treeReconDecay}
\E| \tilde{X}_{v,\ell,s} -p^{s(-1)^{|v|}}|^2 \leq \tilde{C}_1 \exp(-\tilde{C}_2 \ell).
\end{equation}
We now define $\xi^{s,v}$ for $s\in\{-1,+1\}$ and $v\in\dtree$ as the Markov model on the subtree $\dtree^v$ with the same transition matrices but so that the initial distribution at $v$ is given by
\[
\P\left[ \xi^{s,v}_v = 1 \right] = 1 - \P\left[ \xi^{s,v}_v = 1 \right] = q^{s(-1)^{|v|}}.
\]
With this initial distribution $\xi^{s,v}$ is distributed according to the extremal hardcore measure $\hmu^{s(-1)^{|v|}}$ on $\dtree^v$.
Analogously to $\tilde{\xi}$, for a configuration $A$ on $S_{v,\ell}$ and $s\in\{-1,+1\}$ we define the posterior function $h^s_{v,\ell}$ as
\[
h^s_{v,\ell}(A) = \P(\xi^{s,v}_v= 1 | \xi^{s,v}_{v,\ell} = A),
\]
for $s\in\{+,-\}$.  By the definition of conditional probability and the Markov property of $\xi$ and $\tilde{\xi}$,
\begin{align*}
\frac{\P(\xi^{s,v}_v=1 | \xi^{s,v}_{v,\ell} = A)}{\P(\xi^{s,v}_v=0 | \xi^{s,v}_{v,\ell} = A)} &= \frac{\P(\xi^{s,v}_{v,\ell} = A | \xi^{s,v}_v=1)}{ \P(\xi^{s,v}_{v,\ell} = A| \xi^{s,v}_v=0)}\cdot \frac{\P(\xi^{s,v}_v=1)}{\P(\xi^{s,v}_v=0)}\\
&= \frac{\P(\tilde{\xi}^s_{v,\ell} = A | \tilde{\xi}^s_v=1)}{ \P(\tilde{\xi}^s_{v,\ell} = A| \tilde{\xi}^s_v=0)}\cdot \frac{\P(\xi^{s,v}_v=1)}{\P(\xi^{s,v}_v=0)}\\
&= \frac{\P(\tilde{\xi}^s_v=1 | \tilde{\xi}^s_{v,\ell} = A)}{\P(\tilde{\xi}^s_v=0 | \tilde{\xi}^s_{v,\ell} = A)} \cdot \frac{\P(\xi^{s,v}_v=1)\P(\tilde{\xi}^s_v=0)}{\P(\xi^{s,v}_v=0)\P(\tilde{\xi}^s_v=1)}.
\end{align*}
It follows that for some constant $C$ and any configuration $A$ on $S_{v,\ell}$  that
\begin{equation}\label{e:hTildeHComparison}
|h^s_{v,\ell}(A)-q^{s(-1)^{|v|}}|\leq C|\tilde{h}^s_{v,\ell}(A)-p^{s(-1)^{|v|}}|.
\end{equation}
We define
\[
X_{v,\ell,s} = h^s_{v,\ell}(\xi^{s,\rho}_{v,\ell}).
\]
Note that we are taking the posterior function for $v$ but the Markov model rooted at $\rho$ which is a standard object in the recursive analysis of Gibbs measures on trees.  By the Markov property and equation~\eqref{e:hTildeHComparison} we have that
\begin{align*}
\E| X_{v,\ell,s} -q^{s(-1)^{|v|}}|^2 & \leq C\E\left[| \tilde{X}_{v,\ell,s} -p^{s(-1)^{|v|}}|^2 \mid \xi^{s,\rho}_{v}=1\right] + C\E\left[| \tilde{X}_{v,\ell,s} -p^{s(-1)^{|v|}}|^2 \mid \xi^{s,\rho}_{v}=0\right]\\
&\leq 2C \left(\min\{p^{s(-1)^{|v|}},1-p^{s(-1)^{|v|}}\}\right)^{-1} \E | \tilde{X}_{v,\ell,s} -p^{s(-1)^{|v|}}|^2
\end{align*}
and hence we may conclude that there exists constants, $C_1,C_2>0$ such that for all $v\in\dtree$ and $s\in\{-1,+1\}$ we have that
\begin{equation}\label{e:treeReconDecay2}
\E| X_{v,\ell,s} -q^{s(-1)^{|v|}}| \leq C_1 \exp(-C_2 \ell).
\end{equation}

With this result we prove the following stronger version with strong concentration of $X_{\rho,\ell,s}$ around  $q^{s}$.  Similar bounds on this quantity had previously been developed in the colouring model~\cite{BVVW:08} to establish fast mixing of the block dynamics on tree and our proof is partially adapted from theirs.
\begin{lemma}\label{l:reconstructionConcentration}
When $\lambda>\lambda_c$, $q^+ \leq \tfrac35$ and $q^+ q^- < 1/(d-1)$  there exist constants $\zeta_1(\lambda,d),\zeta_2(\lambda,d)>0$ such that for $s\in\{+,-\}$ and for large $\ell$,
\[
\P\left(| X_{\rho,\ell,s} -q^{s}| \geq \exp(-\zeta_1 \ell)\right) \leq \exp\left( - \exp(\zeta_2\ell)\right).
\]
\end{lemma}

Note that the condition $q^+ \leq \tfrac35$ is not necessary but simplifies the proof and holds in the regions of interest.
\begin{proof}
We first observe that the $X_{\rho,\ell,s}$ also satisfy a standard recursive relationship.  If $v\in\dtree$  and $v_1,\ldots,v_{d-1}$ are its children then the standard tree recursion for Gibbs measures of the hardcore model gives,
\begin{equation}\label{e:treeRecursionNoTilde}
X_{v,\ell,s} = \frac{\lambda\prod_{i=1}^{d-1}(1-X_{v_i,\ell,s})}{1+\lambda\prod_{i=1}^{d-1}(1-X_{v_i,\ell,s})}.
\end{equation}
Now note that for any $\delta>0$ there exists $\ell'(d,\lambda,\delta)$ such that for $\ell>\ell'$ we have that $X_{v,\ell,s} < q^+ + \delta$ since this is the case for even conditioning $S_{v,\ell}$ to be all 0 or 1. Now for $0<L<\ell$ write $\mathcal{X}_{L,\ell,s}$ for the vector $\{X_{v,\ell-L,s}:v\in S_{\rho,L}\}$.  Observe that by recursively applying equation~\eqref{e:treeRecursionNoTilde} we can write
\[
X_{\rho,\ell,s} = g(\mathcal{X}_{L,\ell,s}).
\]
Suppose that $\mathcal{X}_{L,\ell,s},\mathcal{X}_{L,\ell,s}'$ are two vectors which are equal except at some $u\in S_{\rho,L}$.  We will now estimate $|g(\mathcal{X}_{L,\ell,s})-g(\mathcal{X}_{L,\ell,s}')|$.  First consider one step of the recursion~\eqref{e:treeRecursionNoTilde} (i.e. the case $L=1$).  Let $u_1,\ldots,u_{d-1}$ be the children of $\rho$ and suppose that $u=u_1$. This implies that
\begin{align*}
&|g(\mathcal{X}_{1,\ell,s})-g(\mathcal{X}_{1,\ell,s}')|
=\left|\frac{1}{1+\lambda\prod_{i=1}^{d-1}(1-X_{u_i,\ell-1,s})}-\frac{1}{1+\lambda\prod_{i=1}^{d-1}(1-X_{u_i,\ell-1,s}')} \right|.
\end{align*}
Now for $\alpha>0$ we have that $\left|\frac{d}{dx} \frac1{1+\alpha x}\right| =\frac{\alpha}{(1+\alpha x)^2}$.  If $x \geq \frac13$ then by a simple optimization we have that $\frac{\alpha}{(1+\frac13\alpha)^2} \leq\frac34$.  Hence if $\ell-L > \ell'(d,\lambda,1/15)$ then
\[
\min\{1-X_{u_i,\ell-1,s},1-X_{u_i,\ell-1,s}'\} > 1-(q^+ + \frac1{15})\geq \frac13.
\]
It follows that
\begin{align*}
|g(\mathcal{X}_{1,\ell,s})-g(\mathcal{X}_{1,\ell,s}')|
&=\left|\int_{1-X_{u,\ell-1,s}}^{1-X_{u,\ell-1,s}'} \frac{\lambda\prod_{i=2}^{d-1}(1-X_{u_i,\ell-1,s})}{(1+ x \lambda\prod_{i=2}^{d-1}(1-X_{u_i,\ell-1,s}))^2}  dx \right|\\
&\leq\left|\int_{1-X_{u,\ell-1,s}}^{1-X_{u,\ell-1,s}'}\frac34 dx \right|= \frac34 \left|X_{u,\ell-1,s} -X_{u,\ell-1,s}' \right|.
\end{align*}
Recursively applying this relation implies that for all $L$ such that  $\ell-L > \ell'(d,\lambda,1/15)$,
\[
|g(\mathcal{X}_{L,\ell,s})-g(\mathcal{X}_{L,\ell,s}')| \leq \left(\frac34\right)^L \left|X_{u,\ell-L,s} -X_{u,\ell-L,s}' \right|.
\]
By the Markov property of the configuration we have that the elements of $\mathcal{X}_{L,\ell,s}$ are conditionally independent given $\xi_{\rho,L}^{s,\rho}$.  Moreover, for $u\in S_{\rho,L}$ the Markov property also implies that $X_{u,\ell,s}$ depends on $\xi_{\rho,L}^{s,\rho}$ only through $\xi^{s,\rho}_{u}$. Since $\P(\xi^{s,\rho}_{u}=1)$ and $\P(\xi^{s,\rho}_{u}=0)$ are strictly bounded away from 0 independent of $L$ we have that by equation~\eqref{e:treeReconDecay2},
\begin{equation}\label{e:onePointCorrelations}
\E\left[| X_{u,\ell,s} -q^{s(-1)^{|u|}} \mid \xi_{\rho,L}^{s,\rho} \right] \leq C_1' \exp(-C_2 \ell).
\end{equation}
Now choose some constant $0<r<1$ such that $r\log(d-1) - C_2(1-r) < r\log(5/4)$ and set $L=\lfloor r\ell\rfloor$.
By Markov's inequality,
\begin{align*}
\P\left( \sum_{u\in S_{\rho,L}} \left| X_{u,\ell-L,s} - q^{s(-1)^{L}}\right| > (\frac54)^L \mid \xi_{\rho,L}^{s,\rho}\right)
&\leq \frac{\E\left( \exp\left(\sum_{u\in S_{\rho,L}} \left| X_{u,\ell-L,s} - q^{s(-1)^{L}}\right|\right) \mid \xi_{\rho,L}^{s,\rho}\right)}{\exp((\frac54)^L)}\\
&\leq \frac{\prod_{u\in S_{\rho,L}} \E \left(\exp\left( \left| X_{u,\ell-L,s} - q^{s(-1)^{L}}\right|\right) \mid \xi_{\rho,L}^{s,\rho}\right)}{\exp((\frac54)^L)}\\
&\leq \frac{\prod_{u\in S_{\rho,L}} \left(1+ eC_1' \exp(-C_2 (\ell-L))\right)}{\exp((\frac54)^L)}\\
&\leq \exp\left((d-1)^L eC_1' \exp(-C_2 (\ell-L)) - (\frac54)^L\right)\\
&\leq \exp\left(- \exp(\zeta_2\ell) \right)
\end{align*}
where the last inequality holds for large $\ell$ when $0<\zeta_2 <r \log\frac54$.  Now if $X_{u,\ell-L,s} = q^{s(-1)^{L}}$ for all $u\in S_{\rho,L}$ then $g(\mathcal{X}_{L,\ell,s})=q^s$ by the standard tree recursions.  By equation~\eqref{e:onePointCorrelations} we have that if
\[
\sum_{u\in S_{\rho,L}} \left| X_{u,\ell-L,s} - q^{s(-1)^{L}}\right| \leq (\frac54)^L
\]
then
\[
| X_{\rho,\ell,s} -q^{s}| \leq O\left((\frac34 \cdot\frac54)^{r\ell}\right)
\]
and so the lemma holds taking $\zeta_1 < r\log \frac{16}{15}$.
\end{proof}

\subsection{The measure on $\sigma_V$}
For compactness of notation we will write the results of this section in terms of the plus phase but the analogous results will hold equally for the minus phase.  Let $\qt_U^+$ denote measure on $\{0,1\}^U$ given by
\[
\qt_U^+(\eta) = \P_{\gt}\left(\sigma_U=\eta \mid Y(\sigma)= + \right).
\]
Lemma~\ref{l:gtExpectedPartition} shows that in expectation at least $\qt_U^+(\eta)$ behaves like $Q_U^+(\eta)$.

The graph $G$ consists of $\gt$ together with a collection of $(d-1)$-ary trees attached to $U$.  Let $I_\eta(G\setminus\gt)$ denote the independent sets on $G\setminus\gt$ which are compatible with the boundary condition $\eta$.  Then the measure on the new part of  $G$ is given by
\[
\P_{G}\left(\sigma_{(G\setminus \gt)\cup U} \mid Y(\sigma)= + \right)  = \frac{\qt_U^+(\sigma_U) \mathbbm{1}_{\{\sigma_{G\setminus \gt}\in I_{\sigma_U}(G\setminus\gt)\}}\lambda^{|\sigma_{G\setminus \gt}|}} {\sum_{\sigma'\in\{0,1\}^{(G\setminus \gt)\cup U}} \qt_U^+(\sigma_U') \mathbbm{1}_{\{\sigma_{G\setminus \gt}'\in I_{\sigma_U'}(G\setminus\gt)\}}\lambda^{|\sigma_{G\setminus \gt}'|}}.
\]
Now since $\qt_U^+(\eta) \approx Q_U^+(\eta)$, at least in expectation, it will also be of interest to consider the measure
\[
P^*(\sigma_{(G\setminus \gt)\cup U}) = \frac{Q_U^+(\sigma_U) \mathbbm{1}_{\{\sigma_{G\setminus \gt}\in I_{\sigma_U}(G\setminus\gt)\}}\lambda^{|\sigma_{G\setminus \gt}|}} {\sum_{\sigma'\in\{0,1\}^{(G\setminus \gt)\cup U}} Q_U^+(\sigma_U') \mathbbm{1}_{\{\sigma_{G\setminus \gt}'\in I_{\sigma_U'}(G\setminus\gt)\}}\lambda^{|\sigma_{G\setminus \gt}'|}}.
\]
The graph $(G\setminus \gt)\cup U$ consists of $(d-1)$-ary trees of depth $2\lfloor\frac{\psi}{2} \log_{d-1}n\rfloor$ rooted at each of the vertices of $V$ and the leaves constitute $U$.  For each $v\in V$ let $T_v$ denote the tree attached to $v$.
\begin{lemma}\label{e:GTreeMeasure}
A configuration $\sigma\in\{0,1\}^{(G\setminus\gt)\cup U}$ distributed according to $P^*$ has the following properties:
\begin{enumerate}
\item The collection of projections $\{\sigma_{T_v}\}_{v \in V}$ are independent.
\item For each $v\in V^\pm$ the measure on $\sigma_{T_v}$ is  given by the projection of $\hmu^\pm$ onto the first $2\lfloor\frac{\psi}{2} \log_{d-1}n\rfloor$ rows of the infinite $(d-1)$-ary tree.
\end{enumerate}
\end{lemma}

\begin{proof}
Since the trees $T_v$ are disconnected and $Q_U^+$ is a product measure the independence of the $\sigma_{T_v}$ is immediate.  Verifying the distribution of $\sigma_{T_v}$ can easily be calculated via direct computation of the measure on the trees.  However, we present a different proof which we feel better illustrates the replica method intuition underlying the result.

Suppose that in the graph $\gt$ we had that $\qt_U^+(\eta) = Q_U^+(\eta)$ holds exactly.
Then the projection of the measure onto $(G\setminus\gt)\cup U$ is exactly given by $P^*$.  Note that this did not depend on the structure of $\gt$ except through $\qt_U^+(\eta)$.  So consider the graph $\tilde{G}^*$ which consists on $2m'$ infinite $(d-1)$-ary trees whose roots we identify with $U$.  Now take the Gibbs  measure $\P_{\tilde{G}^*}$ on configurations on $\tilde{G}^*$ as follows: the measure is a product measure over the different trees and the measure restricted to an individual tree is $\hmu^+$ for trees rooted in $U^+$ and $\hmu^-$ for trees rooted in $U^-$.

Note that with this choice of graph and Gibbs measure the  measure  $\P_{\tilde{G}^*}(\sigma_U \in \cdot)$  is $Q_U^+$.  Now construct $G^*$ by appending trees onto $U$ identically as in the construction of $G$. The resulting graph is a collection of $2m$ disconnected $(d-1)$-ary trees rooted at the vertices of $V$.  The resulting Gibbs measure $\P_{G^*}$  is  a product measure over the trees of $G^*$ and restricted to individual trees it corresponds to the measure $\hmu^+$ for trees rooted in $V+$ and $\hmu^-$ for trees rooted in $V^-$.  By construction the measure on $(G^* \setminus \tilde{G}^*) \cup U$ is identical to $P^*$ which completes the proof.
\end{proof}

Appending trees onto $\gt$ to construct $G$ has the effect of reweighting the projection of the measure on $\{0,1\}^U$.  If we denote
\[
\kappa(\eta) = \sum_{\sigma\in I_\eta(G\setminus\gt)} \lambda^{|\sigma|}
\]
for $\eta\in\{0,1\}^U$ then
\[
\P_{G}\left(\sigma_{U}=\eta \mid Y(\sigma)= + \right)  = \frac{\qt_U^+(\eta) \kappa(\eta)} {\sum_{\eta\in\{0,1\}^{ U}} \qt_U^+(\eta)\kappa(\eta) }.
\]
while
\begin{equation}\label{e:pStarDefn}
\P^*\left(\sigma_{U}=\eta \right)  = \frac{Q_U^+(\eta) \kappa(\eta)} {\sum_{\eta\in\{0,1\}^{ U}} Q_U^+(\eta)\kappa(\eta) }.
\end{equation}
Let $\mathcal{B}$ denote the set of configurations $\eta$ which have a large influence on $\sigma_V$ by
\[
\mathcal{B} = \left\{\eta\in\{0,1\}^U: \sup_{v\in V} |\P_G(\sigma_v=1\mid\sigma_U=\eta) - Q_V^+(\sigma_v=1)|> \exp(-2 \zeta_1 \lfloor\frac{\psi}{2} \log_{d-1}n\rfloor)\right\},
\]
where $\zeta_1$ is as in Lemma~\ref{l:reconstructionConcentration}
Observe that by the Markov property that the distribution of $\sigma_v$ depends only on $\eta$ through the leaves of $T_v$ and hence that $\mathcal{B}$ is independent of $\gt$.  By Lemma~\ref{l:reconstructionConcentration} and a union bound we have that
\begin{equation}\label{e:ProbBadB}
\P^*\left(\sigma_{U}\in \mathcal{B} \right) \leq |V| \exp\left( - \exp\left(2 \zeta_2 \lfloor\frac{\psi}{2} \log_{d-1}n\rfloor\right)\right).
\end{equation}
We are now ready to prove Theorem~\ref{t:Gproperties}.

\begin{proof}(Theorem~\ref{t:Gproperties})

The number of vertices follows immediately from the construction.  First choose $\theta$ small enough so that
\begin{equation}\label{e:thetaCond}
\theta < \min \{\frac{\zeta_1 \psi}{5 \log(d-1)}, \frac{\zeta_2 \psi}{3} \}.
\end{equation}
We will first prove equation~\eqref{e:GpropB}.  Since $\sigma_V$ is conditionally a product measure given $\sigma_U=\eta$ we have that  then for all $\eta'\in\{0,1\}^V$ that,
\begin{align*}
\left|\frac{\P_G\left(\sigma_V=\eta'\mid \sigma_U=\eta\right)}{Q_V^+(\eta')} -1 \right|
\leq \left| \prod_{v\in V} \frac{\P_G\left(\sigma_v=\eta_v'\mid \sigma_U=\eta\right)}{Q_V^+(\sigma_v=\eta_v')} -1 \right|.
\end{align*}
Now if $\eta\in \{0,1\}^U \setminus \mathcal{B}$ then by the definition of $\mathcal{B}$ we have that

\[
\left|\frac{\P_G\left(\sigma_v=\eta_v'\mid \sigma_U=\eta\right)}{Q_V^+(\sigma_v=\eta_v')} -1 \right| \leq O\left(\exp(-2 \zeta_1 \lfloor\frac{\psi}{2} \log_{d-1}n\rfloor)\right).
\]
It follows that if we take $\theta$ according to~\eqref{e:thetaCond}
then for large $n$ and all $\eta\in \{0,1\}^U \setminus \mathcal{B}$ and $\eta'\in\{0,1\}^V$ then,
\begin{equation}\label{e:weakEtaDependence}
\left|\frac{\P_G\left(\sigma_V=\eta'\mid \sigma_U=\eta\right)}{Q_V^+(\eta')} -1 \right| \leq n^{-3\theta},
\end{equation}
since $|V|=O(n^\theta)$.  Hence we have that
\begin{equation}\label{e:lInfinityBound}
\max_{\eta'} \left| \frac{\P_G(\sigma_V=\eta'|Y=+)}{Q_V^+(\sigma_V)} - 1\right| \leq n^{-3\theta} + \frac{\P_G(\sigma_U\in \mathcal{B}\mid Y=+)}{\min_{\eta'} Q^+_V(\eta')}.
\end{equation}
Now since $|V|=O(n^\theta)$ we have the  inequality $\max_{\eta'} \left( Q^+_V(\eta')\right)^{-1}= \exp\left( O(n^\theta)\right)$.  To get a bound on $\P_G(\sigma_U\in \mathcal{B}\mid Y=+)$ recall that
\begin{align}\label{e:ratioBadBRep}
\P_{G}\left(\sigma_{U}=\eta \mid Y(\sigma)= + \right)
&= \frac{\sum_{\eta\in\mathcal{B}}\P_{\gt}\left(\sigma_{U}=\eta \mid Y + \right) \kappa(\eta)} {\sum_{\eta\in\{0,1\}^{ U}} \P_{\gt}\left(\sigma_{U}=\eta \mid Y = +\right) \kappa(\eta) } \nonumber\\
&= \frac{\sum_{\eta\in\mathcal{B}} Z^+_{\gt}(\eta) \kappa(\eta)} {\sum_{\eta\in\{0,1\}^{ U}} Z^+_{\gt}(\eta) \kappa(\eta) }
\end{align}
Now by Theorem~\ref{t:gtWHP}
\begin{equation}\label{e:tGpropProofA}
\P\left(\sum_{\eta\in\{0,1\}^{ U}} Z^+_{\gt}(\eta) \kappa(\eta)  < \frac1{2\sqrt{n}} \E \sum_{\eta\in\{0,1\}^{ U}} Z^+_{\gt}(\eta) \kappa(\eta)  \right) \to 0
\end{equation}
while by Markov's inequality
\begin{equation}\label{e:tGpropProofB}
\P\left(\sum_{\eta\in\mathcal{B}} Z^+_{\gt}(\eta) \kappa(\eta)  > \frac12\sqrt{n} \E \sum_{\eta\in\mathcal{B}} Z^+_{\gt}(\eta) \kappa(\eta)  \right) \to 0.
\end{equation}
Further, recall that by Lemma~\ref{l:gtExpectedPartition} for all $\eta$,
\begin{equation}\label{e:tGpropProofC}
\E Z^+_{\gt}(\eta) = (1+o(1))Q_U^+(\eta)\E Z^+_{\gt},
\end{equation}
and hence by equations~\eqref{e:pStarDefn} and~\eqref{e:ProbBadB}
\begin{align}\label{e:tGpropProofD}
\frac{\sum_{\eta\in\mathcal{B}} \E Z^+_{\gt}(\eta) \kappa(\eta)} {\sum_{\eta\in\{0,1\}^{ U}} \E Z^+_{\gt}(\eta) \kappa(\eta) }
&= (1+o(1))\frac{\sum_{\eta\in\mathcal{B}} \E \qt^+_{U}(\eta) \kappa(\eta)} {\sum_{\eta\in\{0,1\}^{ U}} \E \qt^+_{U}(\eta) \kappa(\eta) } \nonumber\\
&= \P^*\left(\sigma_{U}\in \mathcal{B} \right) \leq |U|\exp\left( - \exp\left(2 \zeta_2 \lfloor\frac{\psi}{2} \log_{d-1}n\rfloor\right)\right).
\end{align}
Combining equations~\eqref{e:ratioBadBRep}, \eqref{e:tGpropProofA}, \eqref{e:tGpropProofB}, \eqref{e:tGpropProofC} and~\eqref{e:tGpropProofD} it follows that for large $n$,
\[
\P_G(\sigma_U\in \mathcal{B}\mid Y=+) \leq n|V| \exp\left( - \exp\left(2 \zeta_2 \lfloor\frac{\psi}{2} \log_{d-1}n\rfloor\right)\right).
\]
Now provided that $\theta$ satisfies~\eqref{e:thetaCond} then for large $n$,
\[
\P_G(\sigma_U\in \mathcal{B}\mid Y=+) \leq \exp(-n^{2\theta}).
\]
Substituting this into~\eqref{e:lInfinityBound}
\[
\max_{\eta'} \left| \frac{\P_G(\sigma_V=\eta'|Y=+)}{Q_V^+(\sigma_V)} - 1\right| \leq 2 n^{-3\theta},
\]
with room to spare for large $n$.  The analogous statement for the minus phase
\[
\max_{\eta'} \left| \frac{\P_G(\sigma_V=\eta'|Y=-)}{Q_V^-(\sigma_V)} - 1\right| \leq 2 n^{-3\theta},
\]
holds similarly and combining the two establishes equation~\eqref{e:GpropB}.

To establish~\eqref{e:GpropA} we will show that with high probability
\[
\frac{Z^+_G}{Z^-_G}>\frac2n, \quad \frac{Z^-_G}{Z^+_G} \geq \frac2n.
\]
By equation~\eqref{e:gtEZb} $\E Z^+_{\gt}=(1+o(1)) \E Z^-_{\gt}$ and \eqref{e:gtEZa} shows that for all $\eta\in\{0,1\}^U$ we have that
$\E Z^+_{\gt}(\eta)=(1+o(1)) \E Z^-_{\gt}(\eta)$.  In particular we have that since
\[
\E Z^\pm_{G} = \sum_{\eta\in\{0,1\}^{ U}} \E Z^\pm_{\gt}(\eta) \kappa(\eta)
\]
and so
\[
\E Z^+_{G} = (1+o(1)) \E Z^-_{G}.
\]
Now by Theorem~\ref{t:gtWHP}
\begin{equation}
\P\left(Z^\pm_{G}  < \frac1{2\sqrt{n}} \E Z^\pm_{G}  \right)=\P\left(\sum_{\eta\in\{0,1\}^{ U}} Z^\pm_{G}(\eta) \kappa(\eta)  < \frac1{2\sqrt{n}} \E \sum_{\eta\in\{0,1\}^{ U}} Z^\pm_{G}(\eta) \kappa(\eta)  \right) \to 0
\end{equation}
while by Markov's inequality
\begin{equation}
\P\left(Z^+_{G}> \frac13\sqrt{n} \E Z^+_{G} \right) \to 0.
\end{equation}
Combining the previous two equations establishes equation~\eqref{e:GpropA} and completes the proof.
\end{proof}

\section{Technical Condition}\label{s:computerAssit}

Our last result is to verify Condition~\ref{cond:technical} in the case that $\lambda=1,d=6$.  With $g_{\alpha,\beta}(\gamma,\delta,\epsilon)=f(\alpha,\beta,\gamma,\delta,\epsilon)$  we use a computer assisted proof to show that $g_{\alpha,\beta}$ attains its unique maximum in the set~\eqref{e:greekRegion} at the point $(\gamma^*,\delta^*,\epsilon^*)=(\alpha^2,\beta^2,\alpha(1-\alpha-\beta))$ for $(\alpha,\beta)$ in a neighbourhood of $(p^-,p^+)$.  These values are approximately  $p^+\approx 0.40831988 , q^- \approx 0.03546955$ (see~\cite{DFJ:02}).  In~\cite{MWW:09}*{Lemma 6.3} it is shown that $g_{\alpha,\beta}$ is maximized as a function of $\epsilon$ by taking
\begin{equation}\label{e:eHatdefn}
\hat{\epsilon}(\alpha,\beta,\gamma,\delta)=\frac12 \left[ 1 +\alpha -\beta  -2\gamma - \sqrt{(1-\alpha-\beta)^2 + 4(\alpha-\gamma)(\beta-\delta)} \right].
\end{equation}
It thus suffices to show that $\hat{g}_{\alpha,\beta}(\cdot,\cdot)=g_{\alpha,\beta}(\cdot,\cdot,\hat\epsilon)$ is maximized at $(\alpha^2,\beta^2)$ for
\[
0\leq \gamma \leq \alpha, \quad 0 \leq \delta \leq \beta.
\]
It is easy to establish that if
\begin{align*}
f^*(\alpha,\beta,\gamma,\delta) = f^1(\alpha,\gamma) + f^2(\beta,\delta)
\end{align*}
where
\begin{align*}
f^1(\alpha,\gamma) &=  H(\alpha)+H_1(\gamma,\alpha)+H_1(\alpha-\gamma,1-\alpha)\\
f^2(\beta,\delta) &= H(\beta) + H_1(\delta,\beta)  + H_1(\beta-\delta,1-\beta)
\end{align*}
then $f\leq f^*$ since $f-f^*$ is the log of  a probability.  Now as a function of $\gamma$, $f_1(\alpha,\gamma)$ is maximized at is maximized at $\gamma=\alpha^2$.  Similarly as a function of $\delta$, $f_2(\beta,\delta)$ is maximized at $\delta=\beta^2$ and is increasing (resp. decreasing) in $\delta$ for $\delta < \beta^2$ (resp. $\delta>\beta^2$). Direct computation then shows that
\begin{align*}
\hat{g}_{p^-,p^+}((p^-)^2,(p^+)^2) &> 1.430 > 1.425 > f^1(p^-,(p^-)^2) + f^2(p^+,0.015)\\
\hat{g}_{p^-,p^+}((p^-)^2,(p^+)^2) &> 1.430 > 1.414 > f^1(p^-,(p^-)^2) + f^2(p^+,0.330)\\
\end{align*}
so for $(\alpha,\beta)$ in a small enough neighbourhood of $(p^-,p^+)$ and $\delta \in [0,\frac{15}{1000}] \cup [\frac{33}{100},\beta]$ and all $0\leq \gamma\leq \alpha$ we have that
\[
\hat{g}_{\alpha,\beta}(\alpha,\beta) > g_{\alpha,\beta}(\gamma,\delta)
\]
so it suffices to consider the set
\begin{equation}\label{e:optimsationArea}
0\leq \gamma \leq \alpha, \quad \frac{15}{1000} \leq \delta \leq \frac{33}{100}
\end{equation}
which we will denote $\Upsilon=\Upsilon(\alpha,\beta)$.  By~\cite{MWW:09}*{Lemma 5.1} the function $\hat{g}_{\alpha,\beta}(\cdot,\cdot,\hat\epsilon)$ has a stationary point at $(\alpha^2,\beta^2)$ so the result will follow by showing that the Hessian matrix $D^2 \hat{g}_{\alpha,\beta}(\cdot,\cdot)$ is negative definite in the region defined by~\eqref{e:optimsationArea}.  This in turn follows as we have that $D^2 g_{p^-,p^+}$ is negative definite at $((p^-)^2,(p^+)^2)$ (via a direct computation) and that
\begin{equation}\label{e:optDeterminant}
\det D^2 \hat{g}_{\alpha,\beta}(\cdot,\cdot)>0
\end{equation}
in the region defined by~\eqref{e:optimsationArea} for $(\alpha,\beta)$ sufficiently close to $(p^-,p^+)$. This is performed using a computer assisted proof.
By~\cite{MWW:09}*{Lemma 6.4} this determinant is given by
\[
\det D^2 \hat{g}_{\alpha,\beta}(\cdot,\cdot,\hat\epsilon)= \left(\frac{\partial f}{\partial^2\gamma} + \frac{\partial \hat\epsilon}{\partial\gamma} \frac{\partial f}{\partial\gamma\partial\epsilon} \right)\left(\frac{\partial f}{\partial^2\delta} + \frac{\partial \hat\epsilon}{\partial\delta} \frac{\partial f}{\partial\delta\partial\epsilon} \right) - \left(\frac{\partial f}{\partial\gamma\partial\delta} + \frac{\partial \hat\epsilon}{\partial\gamma} \frac{\partial f}{\partial\delta\partial\epsilon} \right)^2(\cdot,\cdot,\hat\epsilon).
\]
where the expressions for the partial derivative are given in~\cite{MWW:09}*{Lemma 6.2 and 6.4} as follows,
\begin{align*}
\frac{\partial f}{\partial^2\gamma} &= -\frac{6}{1-2\beta+\delta-\gamma-\epsilon}-\frac{6}{\alpha-\gamma-\epsilon} + \frac{5}{1-2\alpha+\gamma}\\
&\qquad -\frac{6}{\beta-\delta-(\alpha-\gamma-\epsilon)}+\frac{6}{1-\beta-\gamma-\epsilon}+\frac{4}{\alpha-\gamma}-\frac1{\gamma}\\
\frac{\partial f}{\partial\gamma\partial\epsilon} &= -\frac{6}{1-2\beta+\delta-\gamma-\epsilon}-\frac{6}{\alpha-\gamma-\epsilon}  -\frac{6}{\beta-\delta-(\alpha-\gamma-\epsilon)}+\frac{6}{1-\beta-\gamma-\epsilon}\\
\frac{\partial f}{\partial^2\delta} &= -\frac{6}{1-2\beta+\delta-\gamma-\epsilon} + \frac{5}{1-2\beta+\delta}\\
&\qquad -\frac{6}{\beta-\delta-(\alpha-\gamma-\epsilon)} + \frac{4}{\beta-\delta}-\frac1{\delta}\\
\frac{\partial f}{\partial\delta\partial\epsilon} &= \frac{\partial f}{\partial\gamma\partial\delta} =\frac{6}{1-2\beta+\delta-\gamma-\epsilon} + \frac{6}{\beta-\delta-(\alpha-\gamma-\epsilon)}
\end{align*}
and
\begin{align*}
\frac{\partial \hat\epsilon}{\partial\gamma} &= -1 + \frac{\beta-\delta}{\sqrt{(1-\alpha-\beta)^2+4(\alpha-\gamma)(\beta-\delta)}}\\
\frac{\partial \hat\epsilon}{\partial\delta} &= \frac{\alpha-\gamma}{\sqrt{(1-\alpha-\beta)^2+4(\alpha-\gamma)(\beta-\delta)}}.
\end{align*}

Showing that the determinant is always positive is done using a computer assisted proof.  Mathematica can perform interval arithmetic which given a function $h(\cdot)$ and an interval $[x,y]$ will return an interval containing the range of $h([x,y])$.  This gives rigorous upper and lower bounds on the function including rounding in a conservative (i.e. rigorous) manner.  This approach is slightly complicated by the fact that a couple of the terms go to infinity at the boundary of~\eqref{e:optimsationArea}.  We, therefore, do our estimates in a couple of stages.  First let
\[
h_1 = \frac{\partial f}{\partial^2\delta} + \frac{\partial \hat\epsilon}{\partial\delta} \frac{\partial f}{\partial\delta\partial\epsilon}.
\]
Then
\[
\max_{(\gamma,\delta)\in\Upsilon} h_1(p^-,p^+,\gamma,\delta,\hat\epsilon)
\leq \max_{0\leq i \leq 99,1\leq j \leq 32} h_1\left(p^-,p^+, \left[\frac{\alpha i}{100}, \frac{\alpha(i+1)}{100}\right],\left[\frac{j}{100},\frac{j+1}{100}\right],\hat\epsilon\right) < -17.
\]
where the last inequality is derived using mathematica's interval arithmetic.  In particular we can always take $h_1$ to be negative for $(\alpha,\beta)$ close to $(p^-,p^+)$.

We now analyse the term $\frac{4}{\alpha-\gamma} - \frac{\beta-\delta}{\sqrt{(1-\alpha-\beta)^2 + 4(\alpha-\gamma)(\beta-\delta)} } \frac{6}{\alpha-\gamma-\hat\epsilon}$ which appears in
$\left(\frac{\partial f}{\partial^2\gamma} + \frac{\partial \hat\epsilon}{\partial\gamma} \frac{\partial f}{\partial\gamma\partial\epsilon} \right)$.
As $\gamma$ goes to $\alpha$ this term diverges so we need some further estimates on it before applying interval arithmetic analysis.  When $0\leq y \leq \frac54$ we have that
\[
\left(1+\frac25 y\right)^2 \leq 1+y
\]
and so since
\[
\frac{4(\alpha-\gamma)(\beta-\delta)}{(1-\alpha-\beta)^2}\leq \frac{4\alpha\beta}{(1-\alpha-\beta)^2}\leq 0.19
\]
when $(\alpha,\beta)$ is close to $(p^-,p^+)$ we can take
\[
1+\frac25 \cdot \frac{4(\alpha-\gamma)(\beta-\delta)}{(1-\alpha-\beta)^2} \leq \sqrt{1+\frac{4(\alpha-\gamma)(\beta-\delta)}{(1-\alpha-\beta)^2}}.
\]
Rearranging we conclude that
\[
\sqrt{(1-\alpha-\beta)^2 + 4(\alpha-\gamma)(\beta-\delta)} - (1-\alpha-\beta) \leq \frac35 \frac{4(\alpha-\gamma)(\beta-\delta)}{\sqrt{(1-\alpha-\beta)^2 + 4(\alpha-\gamma)(\beta-\delta)}}.
\]
Now using this inequality and plugging in the definition of $\hat\epsilon$ we have that
\begin{align*}
\alpha-\gamma-\hat\epsilon &= \alpha-\gamma-\frac12 \left[ 1 +\alpha -\beta  -2\gamma - \sqrt{(1-\alpha-\beta)^2 + 4(\alpha-\gamma)(\beta-\delta)}\right] \\
&= \frac12 \left[ \sqrt{(1-\alpha-\beta)^2 + 4(\alpha-\gamma)(\beta-\delta)} - (1-\alpha-\beta)\right]\\
&\leq  \frac65 \frac{(\alpha-\gamma)(\beta-\delta)}{\sqrt{(1-\alpha-\beta)^2 + 4(\alpha-\gamma)(\beta-\delta)}}
\end{align*}
and hence we have that
\begin{align*}
\frac{4}{\alpha-\gamma} -\frac{\beta-\delta}{\sqrt{(1-\alpha-\beta)^2 + 4(\alpha-\gamma)(\beta-\delta)} } \frac{6}{\alpha-\gamma-\hat\epsilon}\leq \frac{-1}{\alpha-\gamma}.
\end{align*}
for $(\alpha,\beta)$ in a neighbourhood of $(p^-,p^+)$.  It follows that
\begin{align*}
\frac{\partial f}{\partial^2\gamma} + \frac{\partial \hat\epsilon}{\partial\gamma} \frac{\partial f}{\partial\gamma\partial\epsilon} & \leq -\frac{6}{1-2\beta+\delta-\gamma-\hat\epsilon} + \frac{5}{1-2\alpha+\gamma} -\frac{6}{\beta-\delta-(\alpha-\gamma-\hat\epsilon)}+\frac{6}{1-\beta-\gamma-\hat\epsilon}\\
&\qquad-\frac{1}{\max\{\frac1{10000},\alpha-\gamma\}}
-\frac1{\max\{\frac1{10000},\gamma\}}\\
&+\left(-1 + \frac{\beta-\delta}{\sqrt{(1-\alpha-\beta)^2+4(\alpha-\gamma)(\beta-\delta)}}\right)\\
&\qquad \cdot \left(-\frac{6}{1-2\beta+\delta-\gamma-\hat\epsilon}  -\frac{6}{\beta-\delta-(\alpha-\gamma-\hat\epsilon)}+\frac{6}{1-\beta-\gamma-\hat\epsilon}\right)
\end{align*}
where we denote the right hand side by $\Psi(\alpha,\beta,\gamma,\delta)$.  Denote
\[
\Phi(\alpha,\beta,\gamma,\delta):= \left[\Psi(\alpha,\beta,\gamma,\delta) \left(\frac{\partial f}{\partial^2\delta} + \frac{\partial \hat\epsilon}{\partial\delta} \frac{\partial f}{\partial\delta\partial\epsilon} \right) - \left(\frac{\partial f}{\partial\gamma\partial\delta} + \frac{\partial \hat\epsilon}{\partial\gamma} \frac{\partial f}{\partial\delta\partial\epsilon} \right)^2 \right]
\]
Since $h_1(p^-,p^+,\gamma,\delta,\hat\epsilon)$ is negative throughout $\Upsilon$ we have that for $(\alpha,\beta)$ in a small neighbourhood of $(p^-,p^+)$,
\begin{align*}
\max_{(\gamma,\delta)\in\Upsilon} \det D^2 \hat{g}_{\alpha,\beta} & \geq \max_{(\gamma,\delta)\in\Upsilon} \Phi(\alpha,\beta,\gamma,\delta).
\end{align*}
Now applying a computer assisted proof using interval arithmetics we get that
\begin{align*}
\max_{(\gamma,\delta)\in\Upsilon} \Phi(p^-,p^+,\gamma,\delta) \geq \max_{0\leq i \leq 99,1\leq j \leq 32} \Phi\left(p^-,p^+,\left[\frac{\alpha i}{100}, \frac{\alpha(i+1)}{100}\right],\left[\frac{j}{100},\frac{j+1}{100}\right]\right) > 1500.
\end{align*}
By continuity of $\Phi$ this inequality  also holds for $(\alpha,\beta)$ in a small neighbourhood of $(p^-,p^+)$. This then establishes  Condition~\ref{cond:technical} in the case that $\lambda=1,d=6$.


\begin{bibdiv}
\begin{biblist}


\bib{AchCoj:08}{article}{
  title={{Algorithmic barriers from phase transitions}},
  author={Achlioptas, D.},
  author={Coja-Oghlan, A.},
  booktitle={IEEE 49th Annual IEEE Symposium on Foundations of Computer Science, 2008. FOCS'08},
  pages={793--802},
  year={2008}
}

\bib{Aldous:01}{article}{
  title={{The $\zeta(2)$ limit in the random assignment problem}},
  author={Aldous, D.J.},
  journal={Random Structures and Algorithms}
  volume={18},
  pages={381--418},
  year={2001}
}

\bib{BKMP:05}{article}{
  title={{Glauber dynamics on trees and hyperbolic graphs}},
  author={Berger, N.},
  author={Kenyon, C.},
  author={Mossel, E.},
  author={Peres, Y.},
  journal={Probability Theory and Related Fields},
  volume={131},
  pages={311--340},
  year={2005},
}





\bib{BerFuj:99}{article}{
  title={{On approximation properties of the independent set problem for low degree graphs}},
  author={Berman, P.},
  author={Fujito, T.},
  journal={Theory of Computing Systems},
  volume={32},
  pages={115--132},
  year={1999},
  publisher={Springer}
}


\bib{BST:10}{article}{
  title={{Reconstruction Threshold for the Hardcore Model}},
  author={Bhatnagar, N.},
  author={Sly, A.},
  author={Tetali, P.},
  journal={arXiv:1004.3531},
  year={2010}
}



\bib{BVVW:08}{article}{
  title={{Reconstruction for colorings on trees}},
  author={Bhatnagar, N.},
  author={Vera, J.},
  author={Vigoda, E.},
  author={Weitz, D.},
  journal={To appear in SIAM J. on Discrete Math.},
  year={2008}
}


\bib{BCMR:06}{article}{
  title={{The Kesten-Stigum reconstruction bound is tight for roughly symmetric binary channels}},
  author={Borgs, C.},
  author={Chayes, J.},
  author={Mossel, E.},
  author={Roch, S.},
  booktitle={Foundations of Computer Science, 2006. FOCS'06. 47th Annual IEEE Symposium on},
  pages={518--530},
  year={2006}
}






\bib{DemMon:10}{article}{
  title={{Ising models on locally tree-like graphs}},
  author={Dembo, A.},
  author={Montanari, A.},
  journal={Ann. Appl. Probab},
  volume={20},
  pages={565--592},
  year={2010}
}



\bib{DFJ:02}{article}{
  title={{On Counting Independent Sets in Sparse Graphs}},
  author={Dyer, M.},
  author={Frieze, A.},
  author={Jerrum, M.},
  journal={SIAM Journal on Computing},
  volume={31},
  pages={1527},
  year={2002}
}

\bib{DyGr:00}{article}{
  title={{On Markov Chains for Independent Sets}},
  author={Dyer, M.},
  author={Greenhill, C.},
  journal={Journal of Algorithms},
  volume={35},
  pages={17--49},
  year={2000},
  publisher={Elsevier}
}


\bib{GalTet:06}{article}{
  title={{Slow mixing of Glauber dynamics for the hard-core model on regular bipartite graphs}},
  author={Galvin, D.},
  author={Tetali, P.},
  journal={Random Structures and Algorithms},
  volume={28},
  pages={427--443},
  year={2006},
}

\bib{Georgii:88}{book}{
  title={{Gibbs measures and phase transitions}},
  author={Georgii, H.O.},
  year={1988},
  publisher={Walter de Gruyter}
}



\bib{Greenhill:02}{article}{
  title={{The complexity of counting colourings and independent sets in sparse graphs and hypergraphs}},
  author={Greenhill, C.},
  journal={Computational Complexity},
  volume={9},
  pages={52--72},
  year={2000},
  publisher={Springer}
}


\bib{JanMos:04}{article}{
  title={{Robust reconstruction on trees is determined by the second eigenvalue}},
  author={Janson, S.},
  author={Mossel, E.},
  journal={Annals of Probability},
  volume={32},
  pages={2630--2649},
  year={2004}
}





\bib{JLR:00}{book}{
  title={{Random graphs}},
  author={Janson, S.},
  author={{\L}uczak, T.},
  author={Ruci{\'n}ski, A.},
  year={2000},
  publisher={Wiley}
}






\bib{JerSin:90}{article}{
  title={{Polynomial-time approximation algorithms for the Ising model}},
  author={Jerrum, M. },
  author={Sinclair, A.},
  journal={Automata, Languages and Programming},
  pages={462--475},
  year={1990},
  publisher={Springer}
}


\bib{Kelly:85}{article}{
  title={{Stochastic models of computer communication systems}},
  author={Kelly, FP},
  journal={Journal of the Royal Statistical Society. Series B (Methodological)},
  volume={47},
  number={3},
  pages={379--395},
  year={1985},
  publisher={Royal Statistical Society}
}



\bib{KesSti:66}{article}{
  title={{Additional limit theorems for indecomposable multidimensional Galton-Watson processes}},
  author={Kesten, H.},
  author={Stigum, B.P.},
  journal={The Annals of Mathematical Statistics},
  volume={37},
  pages={1463--1481},
  year={1966}
}








\bib{KMRSZ:07}{article}{
  title={{Gibbs states and the set of solutions of random constraint satisfaction problems}},
  author={Kr\c{z}aka{\l}a, F.},
  author={Montanari, A.},
  author={Ricci-Tersenghi, F.},
  author={Semerjian, G.},
  author={Zdeborov{\'a}, L.},
  journal={Proceedings of the National Academy of Sciences},
  volume={104},
  pages={10318},
  year={2007},
}






\bib{LubVig:99}{article}{
  title={{Fast convergence of the Glauber dynamics for sampling independent sets}},
  author={Luby, M.},
  author={Vigoda, E.},
  journal={Random Structures and Algorithms},
  volume={15},
  pages={229--241},
  year={1999},
  publisher={John Wiley \& Sons}
}


\bib{MarOli:94}{article}{
   author={Martinelli, F.},
   author={Olivieri, E.},
   title={Approach to equilibrium of Glauber dynamics in the one phase
   region. I. The attractive case},
   journal={Comm. Math. Phys.},
   volume={161},
   date={1994},
   number={3},
   pages={447--486},
}









\bib{MezMon:09}{book}{
  title={{Information, physics, and computation}},
  author={Mezard, M.},
  author={Montanari, A.},
  year={2009},
  publisher={Oxford University Press},
  address={USA}
}

\bib{MPV:87}{book}{
  title={{Spin glass theory and beyond}},
  author={Mezard, M.},
  author={Parisi, G.},
  author={Virasoro, M.A.},
  year={1987},
  publisher={World scientific},
  address={Singapore}
}

\bib{Mossel:04}{article}{
  title={{Survey: information flow on trees}},
  author={Mossel, E.},
  booktitle={Graphs, morphisms, and statistical physics: DIMACS Workshop Graphs, Morphisms and Statistical Physics, March 19-21, 2001, DIMACS Center},
  pages={155--170},
  year={2004},
}



\bib{MMS:09}{article}{
  title={{The weak limit of Ising models on locally tree-like graphs}},
  author={Montanari, A.},
  author={Mossel, E.},
  author={Sly, A.},
  journal={ArXiv:0912.0719},
  year={2009}
}



\bib{MosPer:03}{article}{
  title={{Information flow on trees}},
  author={Mossel, E.},
  author={Peres, Y.},
  journal={Annals of Applied Probability},
  volume={13},
  pages={817--844},
  year={2003},
  publisher={Institute of Mathematical Statistics}
}





\bib{MosSly:09}{article}{
  title={{Exact Thresholds for Ising-Gibbs Samplers on General Graphs}},
  author={Mossel, E.},
  author={Sly, A.},
  journal={arXiv:0903.2906},
  year={2009},
}


\bib{MWW:09}{article}{
  title={{On the hardness of sampling independent sets beyond the tree threshold}},
  author={Mossel, E.},
  author={Weitz, D.},
  author={Wormald, N.},
  journal={Probability Theory and Related Fields},
  volume={143},
  pages={401--439},
  year={2009},
  publisher={Springer},
}

\bib{JerSin:89}{article}{
  title={{Approximate counting, uniform generation and rapidly mixing Markov chains}},
  author={Sinclair, A.},
  author={Jerrum, M.}
  journal={Information and Computation},
  volume={82},
  pages={93--133},
  year={1989},
  publisher={Elsevier}
}

\bib{Sokal:01}{article}{
  title={{A personal list of unsolved problems concerning lattice gases and antiferromagnetic Potts models}},
  author={Sokal, A.},
  journal={Markov Process. Related Fields},
  volume={7},
  pages={21--38},
  year={2001},
}

\bib{Sly:08}{article}{
  title={{Uniqueness thresholds on trees versus graphs}},
  author={Sly, A.},
  journal={The Annals of Applied Probability},
  volume={18},
  number={5},
  pages={1897--1909},
  year={2008},
  publisher={Institute of Mathematical Statistics}
}

\bib{Sly:09}{article}{
  title={{Reconstruction for the Potts model}},
  author={Sly, A.},
  booktitle={Proceedings of the 41st annual ACM symposium on Theory of computing},
  pages={581--590},
  year={2009},
}




\bib{Talagrand:06}{article}{
  title={{The Parisi formula}},
  author={Talagrand, M.},
  journal={Annals of Mathematics},
  volume={163},
  pages={221--264},
  year={2006},
}



\bib{Weitz:06}{article}{
 author = {Weitz, Dror},
 title = {Counting independent sets up to the tree threshold},
 conference={
    title={STOC '06: Proceedings of the thirty-eighth annual ACM symposium on Theory of computing},
    date={2006}
 },
 pages = {140--149},
}



\bib{Wormald:99}{article}{
  title={{Models of random regular graphs}},
  author={Wormald, N.C.},
  journal={London Mathematical Society Lecture Note Series},
  pages={239--298},
  year={1999},
}





\end{biblist}
\end{bibdiv}
\end{document}